%% file: main.tex
\documentclass[11pt]{article}

\usepackage[letterpaper,margin=1in]{geometry}
\usepackage[T1]{fontenc}
\usepackage[utf8]{inputenc}
\usepackage{amsmath,amssymb,amsthm,mathtools}
\usepackage{aliascnt}
\usepackage{microtype}
\usepackage{enumitem}
\usepackage{mleftright}
\usepackage{booktabs}
\usepackage{float}
\usepackage{xcolor}
\usepackage[authoryear]{natbib}
\usepackage[colorlinks=true,linkcolor=magenta,citecolor=teal,urlcolor=magenta]{hyperref}
\usepackage{algorithm}
\usepackage{algpseudocode}
\usepackage{tikz}
\usepackage{nicefrac}
\usetikzlibrary{arrows.meta,calc,positioning}
\usepackage{multirow}
\usepackage{bm} %
\usepackage{tabularx}
\usepackage[most]{tcolorbox}

\setlist[itemize]{leftmargin=1.6em,itemsep=0.25em,topsep=0.25em}
\setlist[enumerate]{leftmargin=1.8em,itemsep=0.25em,topsep=0.25em}
\allowdisplaybreaks
\usepackage[capitalize,noabbrev]{cleveref}

\usepackage{xcolor}
\definecolor{myOrange}{RGB}{235, 105, 92}
\definecolor{myMaroon}{RGB}{128, 0, 0}
\definecolor{maroon}{rgb}{0.5,0.0,0.0}
\definecolor{heavymaroon}{rgb}{0.3,0.0,0.0}

\definecolor{mygreen}{RGB}{160, 200, 140}

\usepackage{hyperref}
\hypersetup{
colorlinks=true,
linkcolor=myOrange,
citecolor=myMaroon, filecolor=cyan,
urlcolor=mygreen,
hypertexnames=false
}

\makeatletter
\let\oldalglinenumber\alglinenumber
\renewcommand{\alglinenumber}[1]{%
  \protected@edef\@currentlabel{#1}%
  \oldalglinenumber{#1}%
}
\makeatother

\input{notation}

\title{Multiplicative Optimism for Constant Regret in Games}
\author{
Ashkan Soleymani\\
MIT\\
\texttt{ashkanso@mit.edu}
\and
Georgios Piliouras\\
Google Deepmind\\
\texttt{gpil@google.com}
}
\date{}

\begin{document}
\maketitle

\begin{abstract}
We introduce Multiplicatively Optimistic Regret Matching (\algname{}), an uncoupled learning rule for finite general-sum games.
Under simultaneous full-information self-play, every player achieves
external regret $\mathcal{O}(\sqrt n\log d)$
uniformly over all horizons, using only one-step optimism. The analysis
combines a potential-based regret-matching argument with multiplicative
stability and Hellinger control of strategy movement. A learning-rate
safeguard additionally gives
$\mathcal{O}(\sqrt{T\log d})$ regret in the face of adversarial utilities. 
\end{abstract}

\clearpage
\setcounter{tocdepth}{2}
\tableofcontents
\newpage

\section{Introduction}\label{sec:intro}

A central question in game theory is whether meaningful strategic
behavior can emerge from the independent decisions of learning agents
whose interests need not align. Equilibrium provides a static
description of such behavior, but it does not by itself specify how
agents can reach it through repeated interaction. A learning rule must
therefore serve two purposes. Each player should perform well according
to its own objective, while the joint behavior of the players should
satisfy useful equilibrium guarantees
\citep{fudenberg1998theory,shoham2009multiagent}. We study how quickly
these two goals can be achieved by decentralized learning.

The classical benchmark for stable behavior is Nash equilibrium. At a
Nash equilibrium, no player can improve its expected utility by
unilaterally changing its strategy, and every finite game admits such an
equilibrium \citep{nash1950equilibrium}. For a learner, however,
existence is only part of the problem. A player may not know which
equilibrium to target, how the other players will adapt, or even their
utility functions. This motivates \emph{uncoupled learning}, in which a
player's update does not require access to the objectives of the other
players \citep{hart2003uncoupled}.

Learning in this setting is inherently nonstationary. Even when the
underlying game is fixed, each player's utility vector changes as the
other players change their strategies. These changes in turn affect the
players' future updates, so the environment faced by each learner is
generated by the learning process itself. This feedback is also present
in modern multiagent systems, including self-play and population-based
learning
\citep{silver2017mastering,lanctot2017unified}. Finite games with exact
utility observations remove many statistical and computational
complications of these applications and isolate this strategic source
of nonstationarity.

Regret minimization provides a natural individual performance criterion
for such interactions. A player's external regret compares its
cumulative utility with that of the best fixed action in hindsight,
evaluated against the same sequence of play by the other players.
Sublinear regret means that no fixed action retains a positive
per-round advantage in the long run. Importantly, classical no-regret
algorithms provide this guarantee against arbitrary utility sequences,
without assuming a model of the opponents or that they follow a
prescribed learning rule
\citep{hannan1957approximation,blackwell1956analog,cesa2006prediction}.
The guarantee therefore belongs to each learner individually rather
than only to a coordinated collection of players.

In two-player zero-sum games, this individual guarantee also gives a
direct route to equilibrium. If both players have sublinear external
regret, the saddle-point gap of their average strategies is bounded by
the sum of their average regrets, and the averages approach the set of
Nash equilibria \citep{freund1999adaptive}. This connection underlies
a long line of learning-based methods for solving games, from fictitious
play \citep{robinson1951iterative} to counterfactual regret minimization
and its resulting breakthroughs in large imperfect-information games
\citep{zinkevich2007regret,bowling2015heads,brown2018superhuman,
moravcik2017deepstack,brown2019superhuman}.

The situation is different in general-sum and multiplayer games.
Computing a Nash equilibrium is PPAD-complete already in two-player
general-sum games
\citep{chen2006settling,daskalakis2009complexity}, and multi-agent
learning faces additional obstacles
\citep{hart2003uncoupled,milionis2023impossibility}. More directly for our purposes, small
external regret does not generally imply that the players' average
strategies form an approximate Nash equilibrium. The natural
equilibrium consequence of external regret is instead a
\emph{coarse correlated equilibrium}.

A coarse correlated equilibrium is a distribution over joint play in
which no player can gain by committing in advance to a fixed action
\citep{aumann1974subjectivity,moulin1978strategically,
blum2007external}. This matches external regret exactly. A correlated
equilibrium imposes a stronger condition. After observing its
recommended action, a player should still not benefit from replacing
that recommendation according to a fixed rule. Thus, CCE corresponds to
fixed deviations chosen before the recommendation is known, whereas CE
allows deviations that depend on the recommendation. Stronger notions
such as internal or swap regret lead to correlated-equilibrium
guarantees \citep{hart2000simple,blum2007external}.

This connection is important because it turns an individual learning guarantee into a statement about the joint behavior of all players. Although players randomize independently within each
round, averaging the product distributions generated over time need
not produce another product distribution. This average can therefore
satisfy a CCE guarantee even if the current strategy profile itself
does not converge.

The regret rate translates directly into the rate of convergence to
CCE. With at most $d$ actions and full-information observations, Hedge
guarantees $\mathcal{O}(\sqrt{T\log d})$ regret against an arbitrary, potentially
adversarial utility sequence \citep{freund1997decision,cesa2006prediction}. Averaging
$T$ rounds therefore gives CCE error $\mathcal{O}(\sqrt{\log d/T})$, so reaching
error $\eps$ requires $\mathcal{O}(\log d/\eps^2)$ rounds. If cumulative regret
can instead be bounded uniformly in $T$, the CCE error decreases as
$\mathcal{O}(1/T)$, and an $\eps$-CCE can be reached in $\mathcal{O}(1/\eps)$ rounds up to
the dependence on the game size. Achieving this faster rate under
\emph{self-play} has been a longstanding goal in the study of no-regret
learning in games.

Why might such a faster rate be possible in self-play? Against an
arbitrary sequence of utility vectors, a learner cannot expect one
round to be informative about the next. In a fixed game, however, the
utility sequence is generated by the players' own strategy changes.
An action's expected utility changes only when the opponents change
their mixed strategies. If those strategies move gradually, the utility
sequence becomes predictable from recent observations.

Gradual variation and \emph{optimistic} online learning formalize this idea by
making regret depend on how accurately the next utility vector can be
predicted from the past
\citep{chiang2012online,rakhlin2013online}. The simplest predictor is
the preceding utility vector
\citep{rakhlin2013optimization}. In self-play, however, predictability
is endogenous. A player's prediction is accurate when the other players
move slowly, while their movement is itself determined by their learning
rules. A fast self-play guarantee must therefore make these two effects
reinforce each other. Small strategy movement keeps the utility vectors
predictable, and the optimistic regret bound must in turn provide enough
control on movement to keep the prediction errors small.

\citet{syrgkanis2015fast} made this feedback between predictability and
movement precise through \emph{regret-bounded-by-variation-in-utilities} (RVU)
inequalities. For each player, the regret bound contains a positive term
that depends on changes in the observed utility vectors and a negative
term that depends on changes in the player's own strategy. In a fixed
game, these two quantities are linked because changes in one player's
utilities are caused by movement of the other players. After the
playerwise inequalities are combined, the negative movement terms can
therefore offset part of the positive variation terms. This yields
individual regret of order $T^{1/4}$ in general games and establishes
the basic mechanism behind fast learning in self-play. This was a major
improvement over treating every player's utility sequence as fully
adversarial.

Subsequent work pushed this frontier much further. 
\citet{daskalakis2021near} showed that the same one-step Optimistic
Hedge algorithm achieves
$\mathcal{O}(n\log d\log^4 T)$ individual regret in multiplayer
general-sum games. The algorithm itself remains simple; the main
advance is in the analysis, which shows that self-play generates
considerable higher-order regularity in the utility and strategy
sequences. By controlling higher-order finite differences of these
sequences, they reduce the polynomial dependence on $T$ to a
polylogarithmic one.

A different route was developed by \citet{farina2022near}. Their
Log-Regularized Lifted Optimistic FTRL (LRL-OFTRL) operates in a \emph{lifted
space} in which the relevant regret quantity is nonnegative. This is
useful because external regret itself can be negative, so strong bounds
on the sum of the players' regrets need not translate into comparable
bounds for every player. The lifting avoids this cancellation, while
logarithmic regularization provides the multiplicative stability needed
to control the learning trajectory. The resulting algorithm achieves
$\mathcal{O}(nd\log T)$ individual regret and extends beyond finite
normal-form games to general convex games.

A later line of work addressed negative regret more directly.
\citet{soleymani2025faster} introduced \emph{Cautious Optimism}, which equips
Optimistic Multiplicative Weights Update (OMWU) with an adaptive, non-monotone
learning rate. When a player's regret becomes too negative, the learning
rate is reduced, this \emph{pacing} mechanism limits how much negative regret can offset positive
regret elsewhere in the analysis. This yields
$\mathcal{O}(n\log^2 d\log T)$ individual regret, improving the
dependence on $d$ of LRL-OFTRL and the dependence on $T$ of Optimistic
Hedge. \citet{soleymani2025cautious} extended this idea to a general framework
for optimistic FTRL, allowing different regularizers across players
while retaining an $\mathcal{O}(\sqrt T)$ adversarial regret guarantee.
Together, these works show that fast individual regret requires not
only predictable utilities, but also a mechanism that prevents large
negative regret from undermining playerwise guarantees.

Two recent works, developed independently and concurrently with ours, removed the remaining dependence on the
horizon through \emph{high-order optimism}
\citep{liu2026constant,abbadi2026constant}. Despite being developed
independently, the two algorithms have a closely related structure.
Both use optimistic FTRL on a lifted simplex, where an additional mass
variable turns external regret into a nonnegative lifted regret. Their
entropic regularizers also make this mass affect the sensitivity of the
played strategy, effectively pacing the response even though the
nominal learning rate is fixed. The main difference lies in how high-order optimism is implemented.
ECHO-OFTRL uses a cascade of exponential moving averages to produce a
geometrically stabilized $n$th-order prediction error, while HOOD uses
a discounted $(n+1)$st-order recurrence. These constructions yield
horizon-independent individual regret of
$\mathcal{O}(n^{21}\log^4 d)$ and
$\mathcal{O}(n^3\log^2 d)$, respectively. The results establish that
constant individual regret is possible in general finite games, but
they rely on high-order predictors whose order grows with the number
of players together with a complicated lifted response that adaptively moderates
the sensitivity of the strategy. This leaves open whether the same
horizon independence can be obtained from a substantially simpler
form of optimism.

The two horizon-independent results change the nature of the question.
They show that constant individual regret is possible in general
finite-game self-play, but they obtain it through high-order optimism,
lifted dynamics, and fairly large polynomial dependence on the number
of players. This suggests looking more closely at what is actually
needed for fast no-regret learning.

A particularly appealing target would retain the simplicity of the
earlier optimistic algorithms. The prediction should use only the
previous utility vector, rather than an order that grows with the
number of players, and the update itself should remain easy to
implement. At the same time, removing the dependence on $T$ is only
one part of the rate. The dependence on the size of the game also
matters. Logarithmic dependence on the number of actions is natural
from classical multiplicative-weights guarantees, while the
$\sqrt n$ dependence already appearing in the RVU analysis of
\citet{syrgkanis2015fast} suggests that a linear dependence on the
number of players should not be necessary.

A separate issue is robustness to adversarial utility sequences. Ideally, protection
against arbitrary utility sequences should not require replacing the
self-play algorithm by a different procedure. The adaptive construction
of \citet{daskalakis2021near}, which modifies only the learning rate,
suggests that such protection can be added with a minimal change to the
underlying dynamics. More generally, one would like the resulting
guarantee to have explicit, moderate constants and to come from a proof
whose main mechanism is easy to identify.

These considerations lead to the following question.

\begin{tcolorbox}[
    colback=myMaroon!5,
    colframe=myMaroon,
    boxrule=1.4pt,
    arc=2mm,
    left=-0.5pt,
    right=-0.5pt,
    top=3pt,
    bottom=3pt,
    width=\linewidth,
    halign=center
]
\emph{Can a simple \textbf{one-step optimistic} learner achieve \textbf{horizon-independent} individual regret, \textbf{logarithmic dependence} on the
number of actions \textbf{$d$}, and \textbf{sublinear dependence} on the number of players \textbf{$n$},
while retaining adversarial protection through only a \textbf{learning-rate
safeguard}?}
\end{tcolorbox}

\begin{table}[!t]
\centering
\small
\renewcommand{\arraystretch}{1.10}
\setlength{\tabcolsep}{3.5pt}

\begin{tabularx}{\textwidth}{
    @{}
    >{\raggedright\arraybackslash}p{0.27\textwidth}
    >{\raggedright\arraybackslash}p{0.18\textwidth}
    >{\centering\arraybackslash}p{0.16\textwidth}
    >{\raggedright\arraybackslash}p{0.15\textwidth}
    >{\centering\arraybackslash}X
    @{}
}
\toprule
\textbf{Method}
&
\textbf{Regret in Games}
&
\textbf{Iteration Cost}
&
\textbf{Adversarial Regret}
&
\textbf{Optimism Order}
\\
\midrule

Hedge\\[-2pt]
{\footnotesize\citep{cesa2006prediction}}
&
$\mathcal{O}(\sqrt{T\log d})$
&
$\mathcal{O}(d)$
&
$\mathcal{O}(\sqrt{T\log d})$
&
$0$
\\
\midrule

Optimistic FTRL / OMD\\[-2pt]
{\footnotesize\citep{syrgkanis2015fast}}
&
$\mathcal{O}(\sqrt n\log d\,T^{1/4})$
&
Reg. dep.
&
$\mathcal{O}(\sqrt{T\log d})$
&
$1$
\\
\midrule

Optimistic Hedge\\[-2pt]
{\footnotesize\citep{chen2020hedging}}
&
$\mathcal{O}(n\log^{5/6}d\,T^{1/6})^{\dagger}$
&
$\mathcal{O}(d)$
&
$\mathcal{O}(\sqrt{T\log d})$
&
$1$
\\
\midrule

Optimistic Hedge\\[-2pt]
{\footnotesize\citep{daskalakis2021near}}
&
$\mathcal{O}(n\log d\log^4 T)$
&
$\mathcal{O}(d)$
&
$\mathcal{O}(\sqrt{T\log d})$
&
$1$
\\
\midrule

Clairvoyant MWU\\[-2pt]
{\footnotesize\citep{piliouras2022beyond}}
&
$\mathcal{O}(n\log d)^{\ddagger}$
&
$\mathcal{O}(d)$
&
No guarantee
&
$-$
\\
\midrule

LRL-OFTRL\\[-2pt]
{\footnotesize\citep{farina2022near}}
&
$\mathcal{O}(nd\log T)$
&
$\mathcal{O}(d\log\log T)$
&
$\mathcal{O}(\sqrt{T\log d})$
&
$1$
\\
\midrule

Cautious Optimism\\[-2pt]
{\footnotesize\citep{soleymani2025faster,soleymani2025cautious}}
&
$\mathcal{O}(n\log^2 d\log T)$
&
$\mathcal{O}(d\log\log T)$
&
$\mathcal{O}(\sqrt{T\log d})$
&
$1$
\\
\midrule

ECHO-OFTRL\\[-2pt]
{\footnotesize\citep{liu2026constant}}
&
$\mathcal{O}(n^{21}\log^4 d)$
&
Unknown
&
No guarantee
&
$n$
\\
\midrule

HOOD\\[-2pt]
{\footnotesize\citep{abbadi2026constant}}
&
$\mathcal{O}(n^3\log^2 d)$
&
Unknown
&
$\mathcal{O}(\sqrt{T\log d})$
&
$n+1$
\\
\specialrule{0.9pt}{1pt}{1pt}

\textbf{MORM}\\[-2pt]
{\footnotesize\textbf{[This work]}}
&
$\boldsymbol{\mathcal{O}(\sqrt n\log d)}$
&
$\boldsymbol{\mathcal{O}(d)}$
&
$\boldsymbol{\mathcal{O}(\sqrt{T\log d})}^{\S}$
&
$\boldsymbol{1}$
\\
\bottomrule
\end{tabularx}

\caption{Comparison of no-regret learning algorithms in finite
general-sum games. Here $n$ is the number of players, $T$ the number
of rounds, and $d$ the maximum number of actions. Universal constants
are suppressed. Computational cost is the per-player cost of one
learning update after the utility vector is available and does not
include the cost of computing that vector. Optimism order $0$ denotes
a non-optimistic method, while order $1$ uses one-step prediction.
$\dagger$ applies only to two-player games.
$\ddagger$ Clairvoyant MWU gives the stated bound only on a selected
subsequence of iterates.
$\S$ The adversarial guarantee for Multiplicatively Optimistic Regret
Matching (\algname{}) uses the learning-rate safeguard in
\Cref{app:safeguard}, which leaves the self-play run unchanged.}
\label{tab:regret-comparison}
\end{table}

We answer this question affirmatively. We introduce
\emph{Multiplicatively Optimistic Regret Matching} (\algname{}), a
deterministic uncoupled rule satisfying
\[
    \Reg_i^{(T)}
      \leq
      96\sqrt n\,(2+\log d)
    \qquad
    \text{for every player $i$ and every $T\geq1$}.
\]
The guarantee holds in any fixed $n$-player finite game with utilities
in $[0,1]$ and at most $d\geq2$ actions per player, when players choose
mixed strategies simultaneously and observe their own exact
expected-utility vectors after play. The rule uses the fixed learning
rate $\eta=1/(32\sqrt n)$ and predicts the current centered utility
using only its value from the preceding round. It requires neither
knowledge of the horizon nor high-order prediction. To our knowledge,
this is the first uncoupled algorithm in this setting with
$\mathcal{O}(\sqrt n\log d)$ individual external regret.

For adversarial protection, we use a safeguard that adapts only the
learning rate to the observed dynamics. The response itself, the
cumulative state, and the update rule remain unchanged. Against
arbitrary, possibly adaptive, full-information utility sequences, the
safeguarded rule satisfies
\[
    \Reg_i^{(T)}
      \leq
      96\sqrt n\,(2+\log d)
      +
      21\sqrt{T(2+\log d)}.
\]
Under self-play, the safeguard never changes the learning rate, so the
trajectory is exactly the same as for the fixed-rate algorithm.
\Cref{thm:self-play,thm:adversarial}, together with
\Cref{prop:safeguard-selfplay}, give the precise statements.

\Cref{tab:regret-comparison} summarizes the progression of finite-game
regret bounds. Compared with the
$\mathcal{O}(n\log^2 d\log T)$ guarantee of
\citet{soleymani2025faster}, our bound removes the dependence on the
horizon, improves the action dependence to $\log d$, and reduces the
player dependence from $n$ to $\sqrt n$. Compared with the recent
horizon-independent results
\citep{liu2026constant,abbadi2026constant}, it achieves substantially
smaller dependence on both $n$ and $d$ while using only one-step
optimism.

The regret bound immediately gives an
$\mathcal{O}(\sqrt n\log d/T)$ CCE guarantee for the average
distribution of play. Hence,
$\mathcal{O}(\sqrt n\log d/\eps)$ rounds suffice to reach CCE error
at most $\eps$.

Our starting point is regret matching. Each action is tracked through
its cumulative advantage over the mixed strategy actually played, and
the response is formed from these cumulative advantages
\citep{hart2001general,cesa2003potential}. A potential on this vector is
chosen so that controlling the potential also controls the largest
cumulative advantage, and hence the player's external regret. This potential-based viewpoint gives us room to redesign both the
response and the analysis. Classical regret matching uses the centering
identity to cancel the first-order change of the potential. We keep this
basic structure, but modify the response so that recent information
enters the cancellation in a more useful way. The potential is then
constructed to support the additional stability and curvature
properties needed by that modified update.

Optimism has previously been incorporated into regret matching through
an additive prediction. Predictive regret matching forms its response
from cumulative regret plus a prediction of the next regret vector
before the positive-part operation and normalization
\citep{farina2021faster}. Our first change is to place the prediction
somewhere else. In \algname{}, the cumulative state remains untouched,
and the preceding centered utility instead multiplies the positive
weight assigned to each action. We refer to this as
\emph{multiplicative optimism}. The distinction is small at the level
of the update, but it changes the potential calculation in an important
way. The optimistic correction produces a negative weighted term whose
weights are exactly the derivatives of the potential. Those same
weights can then be used to pay for both the curvature of the potential
and the movement of the strategy.

This viewpoint also changes how the dependence on the number of players
enters the proof. In the standard RVU arguments, including the
$\ell_1$ path-length analysis used in Cautious Optimism, the change in a
player's utility is controlled by the sum of the opponents' strategy
movements. Squaring that sum introduces one factor of $n$, and summing
the playerwise inequalities introduces another. The resulting closure
requires a learning rate of order $1/n$, which is reflected in the
linear player dependence of the final regret bound
\citep{farina2022near,soleymani2025faster,
soleymani2025cautious}. This dependence might therefore appear to be an
unavoidable price of multiplayer interaction.

The second idea is to measure strategy movement in square-root
coordinates (Hellinger distance) instead. Squared Hellinger distance behaves particularly
well for product distributions. The distance between two product
distributions is controlled by the sum of the squared distances between
their factors. Applying this comparison before separating the opponents
removes the first player factor in the usual $\ell_1$ argument. After
summing over players, only one factor of $n$ remains. This is what makes
a learning rate of order $1/\sqrt n$ possible and ultimately gives the
$\sqrt n$ dependence in the regret bound. The improvement comes from matching the geometry used to
measure strategy movement with the product structure of the game.

These two ideas impose concrete requirements on the potential. Its
partial derivatives must stay positive and change in a controlled
multiplicative way, because the response is obtained by normalizing
those derivatives after the optimistic correction. Its curvature must
be controlled by the same derivative weights that appear in the
regret-matching cancellation. Finally, those weights must control
square-root movement even when their total mass becomes very small.
These requirements are stronger than ordinary smoothness, and they are
the reason we do not start from a standard potential and then try to
adapt the analysis around it.

Instead, we construct the potential step by step from the properties
needed by the proof. The scalar building block is linear on positive
inputs, which preserves the regret certificate, and reciprocal on
negative inputs, which keeps small positive weights multiplicatively
stable. A power-norm aggregation then combines the coordinates while
keeping the initial potential logarithmic in the number of actions.
This is where the $\log d$ dependence enters. The same construction
also supplies the weighted curvature and normalization properties
needed for the Hellinger movement bound. The proof overview in
\Cref{sec:overview} develops this construction in the same order in
which the requirements arise.

The resulting analysis remains potential based throughout. The
one-round inequality has the familiar variation-versus-movement shape
of an RVU bound, but the nonnegative quantity carried across
rounds is the potential itself. We therefore do not need to make the
players' ordinary regrets nonnegative or lift the decision space for
that purpose. Under self-play, Hellinger movement controls prediction
error, while the potential controls that same movement. Combining the
two closes the argument and bounds each player's terminal potential,
which then converts directly into regret.

The potential also makes the adversarial safeguard simple. We monitor
the same potential bound that is guaranteed under self-play and reduce
only the learning rate when that bound is violated, following the
monitoring principle of \citet[Appendix~D]{daskalakis2021near}. The
cumulative state is never reset and the algorithm keeps the same form.
Under self-play the threshold is never crossed, so the safeguard is
inactive. Against arbitrary utility sequences, the rate adjustment
keeps the potential under control and recovers the usual
$\sqrt{T}$-type adversarial dependence.

Although the potential is designed around the proof, the resulting
algorithm is simple to implement. In each iteration, the learner
computes one scalar weight for each action from its cumulative centered
utility, multiplies it by the one-step optimistic correction, and
normalizes the resulting vector to obtain the next mixed strategy. Thus,
for a player with at most $d$ actions, each iteration requires
$\mathcal{O}(d)$ arithmetic operations and $\mathcal{O}(d)$ memory. The
stored state consists of two real numbers per action, namely the
cumulative centered utility and the centered utility from the previous
round, i.e., the optimistic prediction. The adversarial safeguard, if used, adds only one scalar learning
rate.

These ingredients play separate roles in the analysis. The
multiplicative correction uses only the previous observation. The
potential gives the stability needed to control Hellinger movement and
obtain the $\sqrt n$ dependence, while its initial value yields the
logarithmic dependence on $d$. The safeguard changes only the learning
rate and leaves the rest of the update unchanged.

The organization of the paper is as follows.
\Cref{sec:related} places our result within the literature on
regret matching, optimistic learning, fast self-play, and stronger
notions of regret. \Cref{sec:algorithm} introduces the learning model, \algname{}, and its
formal guarantees. \Cref{sec:overview} explains how multiplicative
optimism and the potential are constructed from the curvature,
stability, and movement properties needed in the proof.
\Cref{sec:analysis} proves the self-play regret and CCE guarantees, and
\Cref{app:safeguard} establishes the adversarial extension. Throughout
the self-play analysis, players receive exact utility vectors
in a fixed finite game.

\section{Multiplicatively Optimistic Regret Matching and Its Guarantees}\label{sec:algorithm}

We first describe the uncoupled learning setting of learning dynamics and the information available to each player, then introduce the Multiplicatively Optimistic Regret Matching (\algname{}) update and its explicit implementation. We conclude the section by stating its individual-regret and equilibrium guarantees.

\subsection{Finite Games and Uncoupled Learning}\label{sec:game}

Consider a finite $n$-player game. Each player $i\in[n]$ has a nonempty
finite action set $\mathcal A_i$ and mixed-strategy space
$\mathcal X_i=\Delta(\mathcal A_i)$. The utility function of player $i$ is $\mathcal U_i:\prod_{j=1}^n\mathcal A_j\to[0,1]$.
We assume that there is a public bound $d\geq2$ such that
$|\mathcal A_i|\leq d$ for every player $i$. We write
$\vx=(\vx_1,\ldots,\vx_n)$ for a mixed-strategy profile and $\vx_i[a]$
for the probability that player $i$ assigns to action
$a\in\mathcal A_i$. All logarithms are natural.

We study repeated self-play in this fixed game. At each round $t$, all
players simultaneously choose their mixed strategies $\xt_i$. After the
strategies are chosen, player $i$ observes its expected-utility vector
$\nut_i$. For every action $a\in\mathcal A_i$,
\[
    \nut_i[a]
      =\E_{\vec s_{-i}\sim\bigotimes_{j\ne i}\xt_j}
          [\mathcal U_i(a,\vec s_{-i})],
    \qquad a\in\mathcal A_i.
    \numberthis{eq:feedback}
\]
Thus, $\nut_i[a]$ is the expected utility that player $i$ would obtain
by playing action $a$ against the opponents' current mixed strategies.
We assume exact full-information observations, so player $i$ observes
the entire vector $\nut_i$ rather than a single sampled payoff.

The learning dynamics are uncoupled. When choosing $\xt_i$, player $i$
may use its own past observations, its action set, and the public
parameters. It does not need to know the other players' utility
functions, current mixed strategies, or internal learning states.

We define the
centered utility vector
\[
    \ut_i
      =\nut_i-\ip{\xt_i}{\nut_i}\one,
    \qquad
    \vec u_i^{(0)}=\zero.
    \numberthis{eq:centered}
\]
The coordinate $\ut_i[a]\in[-1,1]$ measures the \emph{instantaneous advantage} of action $a$
over the mixed strategy $\xt_i$ played on round $t$. In particular,
$\ut_i[a]>0$ means that action $a$ would have yielded a higher expected
utility than $\xt_i$ against the opponents' round-$t$ mixed strategies.

We accumulate these advantages over time. The cumulative centered
utility available before round $t$ is
\[
    \vec U_i^{(t)}
      =\sum_{s=1}^{t-1}\vec u_i^{(s)},
    \qquad
    \vec U_i^{(1)}=\zero,
    \qquad
    \vec U_i^{(t+1)}=\vec U_i^{(t)}+\ut_i.
    \numberthis{eq:cumulative}
\]
Thus, $\vec U_i^{(t)}$ depends only on observations from rounds
$1,\ldots,t-1$ and is available when player $i$ chooses $\xt_i$.
After round $T$, the vector $\vec U_i^{(T+1)}$ contains exactly the
centered utilities accumulated through round $T$.

By construction, the centered utility has zero expectation under the
played strategy. Moreover, player $i$'s external regret after $T$ rounds
is the largest cumulative centered utility,
\[
    \ip{\xt_i}{\ut_i}=0,
    \qquad
    \Reg_i^{(T)}
      =\max_{a\in\mathcal A_i}\vec U_i^{(T+1)}[a].
    \numberthis{eq:two-identities}
\]
Allowing the comparator to be any fixed mixed strategy gives the same
regret, since its cumulative utility is a convex combination of the
cumulative utilities of the pure actions. Notice that
$\Reg_i^{(T)}$ need not be nonnegative.

\subsection{The Uncoupled Learning Update of \algname{}}

We now define the learning rule of \algname{}. Set
\[
    c \coloneqq 2+\log d,
    \qquad
    \eta \coloneqq \frac{1}{32\sqrt n}.
    \numberthis{eq:parameters}
\]
For each player $i$, define the potential
\[
    \Psi_i(\vec U)
       =c\left(\sum_{a\in\mathcal A_i}
          f\!\left(\frac{\vec U[a]}c\right)^{c-1}
          \right)^{1/(c-1)},
    \qquad
    f(z)=
    \begin{cases}
       (1-z)^{-1}, & z\leq0,\\
       1+z, & z\geq0.
    \end{cases}
    \numberthis{eq:potential}
\]
The potential is defined on vectors of cumulative centered utilities. We write
$\partial_a\Psi_i$ for the derivative of $\Psi_i$ with respect to
coordinate $a$. Every coordinate of $\nabla\Psi_i(\vec U)$ is strictly
positive.

At the beginning of round $t$, player $i$ evaluates the gradient at the
scaled cumulative centered utility $\eta\vec U_i^{(t)}$. The gradient
provides a positive weight for each action. \algname{} then adjusts these
weights using the preceding centered utility $\vec u_i^{(t-1)}$ and
normalizes them to obtain the next mixed strategy,
\[
    \xt_i
    \ \propto\
    \nabla\Psi_i(\eta\vec U_i^{(t)})
    \odot
    \bigl(\one+4\eta\vec u_i^{(t-1)}\bigr).
\]
Here, $\odot$ denotes coordinatewise multiplication, while $\propto$
means that the resulting positive vector is normalized to sum to one.

\begin{algorithm}[H]
\caption{\algname{} for player $i$}
\label{alg:learning}
\begin{algorithmic}[1]
\Require learning rate $\eta$ from \eqref{eq:parameters} and potential $\Psi_i$ from \eqref{eq:potential}.
\State $\vec U_i^{(1)}\gets\zero$, $\vec u_i^{(0)}\gets\zero$.
\For{$t=1,2,\ldots$}
\State $\displaystyle
   \xt_i\ \propto\
   \nabla\Psi_i(\eta\vec U_i^{(t)})
   \odot\bigl(\one+4\eta\vec u_i^{(t-1)}\bigr)$.
\State Play $\xt_i$ and observe $\nut_i$.
\State $\ut_i\gets\nut_i-\ip{\xt_i}{\nut_i}\one$.
\State $\vec U_i^{(t+1)}\gets\vec U_i^{(t)}+\ut_i$.
\EndFor
\end{algorithmic}
\end{algorithm}

\begin{remark}[Connection to Optimistic MWU] \label{remark:OMWU}
Multiplicative optimism is not specific to the potential used by
\algname{}. With the log-sum-exp potential associated with negative
entropy, the same construction gives
\[
    x_i^{(t)}[a]
      \propto
      \exp\!\left(\eta\vec U_i^{(t)}[a]\right)
      \bigl(1+4\eta u_i^{(t-1)}[a]\bigr).
\]
Thus, the underlying response is MWU, with the previous utility entering
as a multiplicative optimistic correction. Standard Optimistic MWU, or Optimistic Hedge, instead uses an
exponential correction
\citep{rakhlin2013optimization,daskalakis2021near}. In our notation,
its update can be written as
\[
    x_i^{(t)}[a]
      \propto
      \exp\!\left(\eta\vec U_i^{(t)}[a]\right)
      \exp\!\left(\eta u_i^{(t-1)}[a]\right).
\]
Since $\exp(z)=1+z+\mathcal{O}(z^2)$, the correction used by
\algname{} is a first-order analogue of the standard optimistic
correction. The potential in this paper is chosen differently in order
to obtain the curvature and stability properties needed for our
analysis.
\end{remark}

\paragraph{Explicit implementation.}
The gradient of $\Psi_i$ has a closed form, so the update does not
require differentiation at runtime. Define the scalar weight function
\[
    a_c(z)=
    \begin{cases}
       (1-z/c)^{-c}, & z\leq0,\\
       (1+z/c)^{c-2}, & z\geq0.
    \end{cases}
    \numberthis{eq:weight}
\]
Differentiating \eqref{eq:potential} gives
\[
    \partial_a\Psi_i(\vec U)
       =\left(\frac{\Psi_i(\vec U)}c\right)^{2-c}
          a_c(\vec U[a]).
    \numberthis{eq:gradient}
\]
The factor $(\Psi_i(\vec U)/c)^{2-c}$ is the same for every action and
therefore disappears when the weights are normalized. Consequently, the
update can be implemented directly as
\[
    \xt_i[a]\ \propto\
       a_c\bigl(\eta\vec U_i^{(t)}[a]\bigr)
       \bigl(1+4\eta\vec u_i^{(t-1)}[a]\bigr),
    \qquad
    a\in\mathcal A_i.
    \numberthis{eq:explicit-rule}
\]

Thus, each action $\xt_i[a]$ receives a weight determined by two quantities. The
first depends on its cumulative centered utility $\vec U_i^{(t)}[a]$, while the second
depends on its centered utility in the preceding round $\vec u_i^{(t-1)}[a]$. Computing the
strategy requires only coordinatewise powers, multiplication, and a
single normalization. So, a round costs $\mathcal O(|\mathcal A_i|)$ time and $\mathcal O(|\mathcal A_i|)$ memory since the learner stores $\vec U^{(t)}_i$, $u^{(t-1)}_i$, and one scalar learning rate.
In particular, a player implementing the rule of \algname{} for self-play never has to evaluate $\Psi_i$ itself.

\begin{remark}
Since $u_i^{(t-1)}[a]\in[-1,1]$  and
$4\eta=1/(8\sqrt n)\leq1/8$, every optimistic correction satisfies
$1+4\eta u_i^{(t-1)}[a]\in[7/8,9/8]$.
Together with the strict positivity of every coordinate of
$\nabla\Psi_i$, this ensures that the normalization in
\eqref{eq:explicit-rule} is well defined on every round and that
$x_i^{(t)}$ has full support. On the first round,
$\vec U_i^{(1)}=\zero$ and $u_i^{(0)}=\zero$. Since $a_c(0)=1$, all
actions receive the same weight, so $x_i^{(1)}$ is uniform. The same
conclusions hold under the learning-rate safeguard, since the safeguard
can only decrease the learning rate.
\end{remark}

\begin{remark}
The factor
$a_c(\eta\vec U_i^{(t)}[a])$ captures the cumulative performance of
action $a$. It is increasing in $\vec U_i^{(t)}[a]$, so actions with
larger cumulative advantage receive larger weight. For actions with
large negative cumulative advantage, this weight approaches zero, but
it remains strictly positive at every finite input. The second factor,
$1+4\eta\vec u_i^{(t-1)}[a]$, provides the optimistic correction. It
increases the weight of an action that performed better than the
player's mixed strategy in the preceding round and decreases the weight
of an action that performed worse. Thus, the first factor summarizes
past performance over all previous rounds, while the second reacts to
the most recent observation. 
\end{remark}

\subsection{Regret and Equilibrium Guarantees}

We now state the main guarantees of \algname{}. Under self-play, every
player has external regret bounded independently of the time horizon. The
bound grows sublinearly with the number of players, as $\sqrt n$, and only
logarithmically with the maximum number of actions, as $\log d$.
A learning-rate safeguard, which only decreases the learning rate
adaptively (and keeps the update rule the same), also gives the optimal adversarial guarantee in the face of adversarial utilities. We first summarize the two results.

\begin{theorem}[Regret bounds for \algname{}]\label{thm:main}
Consider a fixed $n$-player finite game with at most $d\geq2$ actions
per player and utilities in $[0,1]$. If all players $i\in[n]$ follow
\algname{} (\Cref{alg:learning}) with $c=2+\log d$ and learning rate
$\eta=1/(32\sqrt n)$, observing their exact expected-utility vectors
after each round, then every player $i\in[n]$ has external regret
\[
    \Reg_i^{(T)}\leq \mathcal{O}(\sqrt n\log d).
\]
Moreover, \algname{} is adaptive to adversarial utilities through the
learning-rate safeguard in \Cref{app:safeguard}, which only decreases
the learning rate adaptively and leaves the rest of the update unchanged. The safeguarded rule for any individual
player $i$ guarantees
\[
    \Reg_i^{(T)}\leq \mathcal{O}(\sqrt n\log d+\sqrt{T\log d})
\]
against arbitrary, possibly adaptive, utility vectors
$\nut_i\in[0,1]^{|\mathcal A_i|}$. Both bounds hold uniformly over all $T\geq1$.
\end{theorem}

The precise self-play and adversarial bounds are proved in
\Cref{thm:self-play} and \Cref{thm:adversarial}, respectively. Together,
these results establish \Cref{thm:main}. The self-play bound also yields
convergence of the time-average of the played distributions to the set
of coarse correlated equilibria (CCE).

\begin{corollary}[CCE]\label{cor:cce}
Under the assumptions of \Cref{thm:self-play}, the distribution
\[
    \overline\sigma^{(T)}
      =\frac1T\sum_{t=1}^T\bigotimes_{j=1}^n\xt_j
    \numberthis{eq:average}
\]
is a $96\sqrt n\,(2+\log d)/T$-coarse correlated equilibrium. Hence, no
player can improve its expected utility by more than
$96\sqrt n\,(2+\log d)/T$ by committing instead to any fixed action.
\end{corollary}

\Cref{sec:overview} explains the main ideas behind the construction and
analysis of \algname{}. \Cref{sec:analysis} proves the self-play regret
bound and the equilibrium guarantee. \Cref{app:safeguard} defines the
learning-rate safeguard and proves the adversarial guarantee.

\section{Proof Overview}\label{sec:overview}

The construction of \algname{} starts from the standard regret-matching
argument \citep{hart2001general,cesa2003potential}. Regret matching
chooses the strategy so that the linear change of its potential cancels
exactly, but it still pays a positive quadratic remainder on every
round. Our first step is to modify the response so that this linear term
becomes negative when the current centered utility is predictable from
the previous round. Before choosing the potential, we compare the usual
$\ell_1$ movement argument with a squared Hellinger comparison of the
players' product distributions. The latter avoids an extra player
factor in the prediction-error bound. We then construct $\Psi_i$ so
that its negative weighted square controls this same squared Hellinger
movement. Combining the potential and prediction-error estimates gives
the self-play regret bound.

Fix a player $i$. Throughout the single-player calculations below, we
omit the player index from vectors and potentials, and all sums over
actions are over $\mathcal A_i$. A vector $\vec U$ without a round
superscript denotes an arbitrary input to the potential. Along the
learning trajectory, the input to $\Psi$ is
$\eta\vec U^{(t)}$.

\subsection{The Classical Regret-matching Analysis}\label{sec:rm}

The external-regret version of regret matching chooses
\[
    \xt[a]\ \propto\
       \left[\vec U^{(t)}[a]\right]_+,
    \qquad [z]_+=\max\{z,0\},
\]
using any distribution if all weights are zero. Consider its quadratic
potential
\[
    \Phi(\vec U)=\frac12\sum_a[\vec U[a]]_+^2.
\]
Since $\partial_a\Phi(\vec U)=[\vec U[a]]_+$, regret matching obtains
$\xt$ by normalizing $\nabla\Phi(\vec U^{(t)})$. The centering identity in
\eqref{eq:two-identities} gives
\[
    \left\langle
       \nabla\Phi(\vec U^{(t)}),\ut
    \right\rangle=0.
\]
This identity holds for every utility vector $\nut$ observed after play.
When $\nabla\Phi(\vec U^{(t)})=\zero$, it holds regardless of the chosen
strategy $\xt$.

The derivative of $[z]_+^2/2$ is $1$-Lipschitz, so
$[z+h]_+^2/2\leq[z]_+^2/2+[z]_+h+h^2/2$. Summing over actions gives
\[
    \Phi(\vec U^{(t+1)})-\Phi(\vec U^{(t)})
    \leq
       \left\langle\nabla\Phi(\vec U^{(t)}),\ut\right\rangle
       +\frac12\sum_a\ut[a]^2
       \leq\frac d2.
\]
Telescoping from $\vec U^{(1)}=\zero$ yields
$[\Reg^{(T)}]_+^2/2\leq\Phi(\vec U^{(T+1)})\leq dT/2$, and therefore
$\Reg^{(T)}\leq\sqrt{dT}$ \citep{hart2001general,cesa2003potential}. The linear term
cancels exactly, but the curvature term $\tfrac12\sum_a\ut[a]^2$
still costs up to $d/2$ on each round. We seek a response that makes the
linear term negative enough to offset this cost when the centered utility
$\ut$ is predictable.

\subsection{The Multiplicative Correction}\label{sec:optimism}

We now ask how to modify the regret-matching response so that the
first-order change of the potential can offset its second-order
remainder. For the moment, leave the potential $\Psi$ unspecified,
except that its partial derivatives are strictly positive.

Recall that $\vec U^{(t+1)}=\vec U^{(t)}+\ut$. Hence, over one round,
the input to the potential moves from $\eta\vec U^{(t)}$ to
$\eta\vec U^{(t)}+\eta\ut$. Taylor's theorem gives
\[
    \Psi(\eta\vec U^{(t+1)})-\Psi(\eta\vec U^{(t)})
    =
    \eta\left\langle
       \nabla\Psi(\eta\vec U^{(t)}),\ut
    \right\rangle
    +\text{second-order remainder}.
\]
The first term describes the linear change of the potential and, as in
regret matching, depends directly on the response chosen by the player.
Our goal is to choose the response so that this term becomes negative
enough to compensate for the second-order remainder.

To make this possible, we construct $\Psi$ so that its curvature is
controlled by its own gradient coordinates. This is related in spirit
to the intrinsic Lipschitz viewpoint of
\citet{soleymani2025cautious}, where local variation is controlled using
the underlying geometry. Here, the property we need is a direct
comparison between the Hessian and the gradient weights. In particular,
along the update segment we establish
\[
    {\ut}^\top\nabla^2\Psi(\vec V)\ut
    \leq
    \sum_a \partial_a\Psi(\vec V)\ut[a]^2.
\]
We also show that the gradient coordinates change only
multiplicatively as $\vec V$ moves from $\eta\vec U^{(t)}$ to
$\eta\vec U^{(t)}+\eta\ut$. Combining these two properties with the
integral form of the Taylor remainder gives
\[
    \Psi(\eta\vec U^{(t+1)})-\Psi(\eta\vec U^{(t)})
    \leq
    \eta\left\langle
       \nabla\Psi(\eta\vec U^{(t)}),\ut
    \right\rangle
    +\frac{2\eta^2}{3}\sum_a
       \partial_a\Psi(\eta\vec U^{(t)})\ut[a]^2.
    \numberthis{eq:taylor}
\]
We prove this estimate formally in \Cref{lem:taylor}. The key feature of
\eqref{eq:taylor} is that the second-order term is not bounded simply by
a constant multiple of $\norm{\ut}_2^2$. Instead, each $\ut[a]^2$ is
weighted by $\partial_a\Psi(\eta\vec U^{(t)})$. These are precisely the
gradient weights that we will use to construct the mixed strategy. We
can therefore try to make the first-order term
$\eta\langle\nabla\Psi(\eta\vec U^{(t)}),\ut\rangle$ produce a negative
square with the same weights.

The quadratic regret-matching potential $\Phi$ from \Cref{sec:rm} does
not have this gradient-weighted curvature property. At
$\vec U=\zero$, we have $\nabla\Phi(\zero)=\zero$, so a bound of the
desired form would force the second-order remainder to vanish.
Nevertheless, increasing any coordinate from zero in a positive
direction increases $\Phi$ quadratically. Thus, the curvature of $\Phi$
does not decrease together with its gradient weights. We therefore need
a potential whose curvature can be controlled by the same weights that
define the strategy as in \eqref{eq:taylor}.

\subsubsection{An Ideal Multiplicative Correction}

The form of \eqref{eq:taylor} suggests how the response should be
chosen. Suppose temporarily that the current centered utility $\ut$
were known before the player chooses $\xt$. Consider the response
\[
    \xt[a] \ \propto\
      \partial_a\Psi(\eta\vec U^{(t)})
      \bigl(1+4\eta\ut[a]\bigr).
\]
Thus, we start from the gradient weight
$\partial_a\Psi(\eta\vec U^{(t)})$ and multiply it by a correction that
depends on the current advantage $\ut[a]$.

Let
$
    Z = \sum_b \partial_b\Psi(\eta\vec U^{(t)})
      \bigl(1+4\eta\ut[b]\bigr)
$
be the normalizing constant. Then
\[
    \xt[a] = \frac{
        \partial_a\Psi(\eta\vec U^{(t)})
        \bigl(1+4\eta\ut[a]\bigr)}
      {Z}.
\]
Since the centered utility satisfies $\ip{\xt}{\ut}=0$, we have
\[
    0 = \sum_a\xt[a]\ut[a] =
      \frac1Z\sum_a
      \partial_a\Psi(\eta\vec U^{(t)})
      \bigl(1+4\eta\ut[a]\bigr)\ut[a].
\]
Multiplying by $Z$, expanding, and rearranging gives
\[
    \sum_a
       \partial_a\Psi(\eta\vec U^{(t)})\ut[a] =
    -4\eta\sum_a \partial_a\Psi(\eta\vec U^{(t)})\ut[a]^2.
\]
Therefore, the first-order term in \eqref{eq:taylor} becomes
\[
    \eta\left\langle
       \nabla\Psi(\eta\vec U^{(t)}),\ut
    \right\rangle = -4\eta^2\sum_a
       \partial_a\Psi(\eta\vec U^{(t)})\ut[a]^2.
\]

This is the key observation. The negative first-order term uses exactly
the same weights $\partial_a\Psi(\eta\vec U^{(t)})$ as the positive
second-order term in \eqref{eq:taylor}, but with a larger coefficient.
Thus, if the current centered utility were known in advance, the
first-order decrease would more than compensate for the Taylor
remainder.

Of course, this response cannot be implemented because $\ut$ is only
known after $\xt$ is chosen. Moreover, $\ut$ depends on $\xt$ through
the centering term $\ip{\xt}{\nut}\one$, while $\nut$ depends on the
opponents' current mixed strategies. We therefore use this calculation
only to identify the ideal multiplicative correction
$1+4\eta\ut[a]$.

The executable algorithm replaces the unavailable $\ut$ by the
preceding centered utility $\vec u^{(t-1)}$, which is known before round
$t$. The next step shows that this replacement preserves the useful
negative weighted square up to an error controlled by
$\ut-\vec u^{(t-1)}$.

\subsubsection{From the Ideal Correction to an Optimistic Update}

The ideal correction derived above depends on the current centered
utility $\ut$, which is not available when the player chooses $\xt$.
We therefore predict $\ut$ using the preceding centered utility
$\vec u^{(t-1)}$. This use of a prediction for the current utility is
the \emph{optimistic step} in the update. Substituting this prediction into
the ideal correction gives
\[
    \xt[a]\ \propto\
       \partial_a\Psi(\eta\vec U^{(t)})
       \bigl(1+4\eta\vec u^{(t-1)}[a]\bigr).
    \numberthis{eq:gradient-rule}
\]

This is precisely the update used by \algname{}. It has the same form
as the ideal response, but replaces the unavailable current advantage
$\ut[a]$ by the optimistic prediction $\vec u^{(t-1)}[a]$.

The same normalization argument as before now produces a cross term
between the predicted and realized centered utilities. In particular,
using $\ip{\xt}{\ut}=0$ gives
\[
    \eta\left\langle
       \nabla\Psi(\eta\vec U^{(t)}),\ut
    \right\rangle
    =
    -4\eta^2\sum_a
       \partial_a\Psi(\eta\vec U^{(t)})
       \vec u^{(t-1)}[a]\ut[a].
\]
For the ideal correction, the corresponding product was $\ut[a]^2$.
The price of using the prediction $\vec u^{(t-1)}[a]$ is therefore
determined by how close it is to $\ut[a]$.

To make this precise, for every action $a$,
\[
    -4\vec u^{(t-1)}[a]\ut[a]
    \leq
    2\bigl(\ut[a]-\vec u^{(t-1)}[a]\bigr)^2
    -2\ut[a]^2.
\]
Indeed, the difference between the right-hand side and the left-hand
side is $2\vec u^{(t-1)}[a]^2\geq0$. Thus, the cross term decomposes
into two useful pieces. The first is a positive squared prediction
error, while the second retains a negative multiple of $\ut[a]^2$.

Substituting this inequality into \eqref{eq:taylor} gives
\[
 \hspace{-0.4cm}
    \Psi(\eta\vec U^{(t+1)})-\Psi(\eta\vec U^{(t)})
    & \leq
    2\eta^2\sum_a
       \partial_a\Psi(\eta\vec U^{(t)})
       \bigl(\ut[a]-\vec u^{(t-1)}[a]\bigr)^2
    -\frac{4\eta^2}{3}\sum_a
       \partial_a\Psi(\eta\vec U^{(t)})\ut[a]^2 \\
    & \leq 2\eta^2\sum_a
       \partial_a\Psi(\eta\vec U^{(t)})
       \bigl(\ut[a]-\vec u^{(t-1)}[a]\bigr)^2
    -\eta^2\sum_a
       \partial_a\Psi(\eta\vec U^{(t)})\ut[a]^2.
    \numberthis{eq:one-round}
\]

Inequality \eqref{eq:one-round} captures the role of optimism in the
update. When the prediction is accurate, the positive term involving
$\ut-\vec u^{(t-1)}$ is small, while the negative weighted square
remains. Under perfect prediction,
$\ut=\vec u^{(t-1)}$, the prediction-error term vanishes and the
potential decreases unless $\ut=\zero$.

More generally, the potential can increase only through the squared
prediction error, while the same gradient weights continue to multiply
the negative term involving $\ut[a]^2$. This matching of weights is the
reason for introducing the prediction multiplicatively through
$\one+4\eta\vec u^{(t-1)}$. The coefficient $4$ leaves enough slack for
a negative weighted square to remain after accounting for the Taylor
remainder.

\subsection{Measuring Strategy Movement with Hellinger Distance}
\label{sec:hellinger}

The positive term in \eqref{eq:one-round} is driven by the prediction
error $\ut_i-\vec u_i^{(t-1)}$. In a fixed game, this error is caused
by changes in the players' mixed strategies. We therefore need to
control one-round strategy movement, measured by a squared distance such
as
\[
    \norm{\xt_i-\vx_i^{(t-1)}}_1^2
    \qquad\text{or}\qquad
    \norm{\sqrt{\xt_i}-\sqrt{\vx_i^{(t-1)}}}_2^2.
\]
As we show next, the choice of distance matters for the dependence on
the number of players. The usual $\ell_1$ comparison loses an extra
player factor, while squared Hellinger distance avoids this loss.

\subsubsection{Standard $\ell_1$ Comparison Loses a Player Factor}

To see the issue, first consider the uncentered utility vector $\nut_i$.
By multilinearity of expected utility in a normal-form game,
$\nut_i[a]$ is the expectation of the fixed payoff function
$\mathcal U_i(a,\cdot)$ under the opponents' product distribution.
Since $\mathcal U_i(a,\cdot)\in[0,1]$,
\[
    \norm{\nut_i-\vec\nu_i^{(t-1)}}_\infty
    \leq
    \left\|
       \bigotimes_{j\ne i}\xt_j
       -
       \bigotimes_{j\ne i}\vx_j^{(t-1)}
    \right\|_1.
\]
A telescoping comparison of the two product distributions then gives
\[
    \left\|
       \bigotimes_{j\ne i}\xt_j
       -
       \bigotimes_{j\ne i}\vx_j^{(t-1)}
    \right\|_1
    \leq
    \sum_{j\ne i}
       \norm{\xt_j-\vx_j^{(t-1)}}_1.
\]

Because the prediction-error term is squared, by Cauchy–Schwarz, this gives
\[
    \norm{\nut_i-\vec\nu_i^{(t-1)}}_\infty^2
    \leq
    (n-1)
    \sum_{j\ne i}
       \norm{\xt_j-\vx_j^{(t-1)}}_1^2. \numberthis{eq:extra_n}
\]
This is the key game-level step in the RVU analysis. Under self-play,
predictability of the utility vectors is controlled by the players'
strategy movement, thereby connecting utility variation to the
individual regret bounds through the accumulated squared path length.
The first factor of order $n$ above comes from converting the square of
a sum into a sum of squares. Summing the prediction-error bounds over
players introduces another factor of order $n$, since each player's
movement affects the utilities of every other player. These two factors
force the learning rate to scale as $O(1/n)$ in standard RVU-based
analyses of regularized learning in games
\citep{syrgkanis2015fast,farina2022near,anagnostides2022last,anagnostides2022uncoupled,soleymani2025faster,soleymani2025cautious}.

Our goal is to avoid the first of these two $n$ factors.

\subsubsection{Squared Hellinger Movement}

We instead measure movement using
\[
    \norm{\sqrt{\vx'}-\sqrt{\vx}}_2^2,
\]
where square roots are taken coordinatewise. Up to a conventional factor
of $1/2$, this is squared Hellinger distance.

The key property is that squared Hellinger distance behaves additively
across product distributions. In particular,
\[
    \left\|
       \sqrt{\bigotimes_{j\ne i}\xt_j}
       -
       \sqrt{\bigotimes_{j\ne i}\vx_j^{(t-1)}}
    \right\|_2^2
    \leq
    \sum_{j\ne i}
       \norm{\sqrt{\xt_j}-\sqrt{\vx_j^{(t-1)}}}_2^2.
\]
The crucial difference from the $\ell_1$ comparison is that the
right-hand side is already a sum of squared movements. No additional
Cauchy--Schwarz step, and hence no additional factor of $n$ as in \eqref{eq:extra_n}, is needed.

By standard properties of Hellinger distance, changes in the expectation
of a $[0,1]$-valued function are controlled by the square-root distance
between the underlying distributions. Applying this to
$\mathcal U_i(a,\cdot)$ under the opponents' product distributions on
rounds $t$ and $t-1$ gives, for every $a\in\mathcal A_i$,
\[
    \left|
       \nut_i[a]-\vec\nu_i^{(t-1)}[a]
    \right|
    \leq
    \left\|
       \sqrt{\bigotimes_{j\ne i}\xt_j}
       -
       \sqrt{\bigotimes_{j\ne i}\vx_j^{(t-1)}}
    \right\|_2.
\]
Taking the maximum over actions, squaring, and using the product
comparison above yields
\[
    \norm{\nut_i-\vec\nu_i^{(t-1)}}_\infty^2
    \leq
    \sum_{j\ne i}
       \norm{\sqrt{\xt_j}-\sqrt{\vx_j^{(t-1)}}}_2^2.
\]
These properties of Hellinger distance are proved in \Cref{app:game}.

Our analysis uses centered utilities $\ut_i$ rather than the uncentered vectors
$\nut_i$. Centering introduces an additional dependence on player $i$'s
own strategy, but only changes the constant in the prediction-error
bound. Combining the opponents' product-distribution comparison with
the centering term gives
\[
    \norm{\ut_i-\vec u_i^{(t-1)}}_\infty^2
    \leq
    5\sum_{j=1}^n
       \norm{\sqrt{\xt_j}-\sqrt{\vx_j^{(t-1)}}}_2^2,
    \qquad t\geq2.
    \numberthis{eq:game-error}
\]
\Cref{lem:game} proves this bound in detail. The important point is that
the prediction error is controlled by a sum of squared strategy
movements without an additional player factor inside the inequality.
Summing \eqref{eq:game-error} over players therefore introduces only
one factor of $n$.

\subsubsection{What the Potential Must Control}

The argument above for efficient regret bounds in the self-play identifies the movement measure we need.
To close the self-play analysis, the negative gradient-weighted square
in \eqref{eq:one-round} must control
\[
    \norm{\sqrt{\vx_i^{(t+1)}}-\sqrt{\xt_i}}_2^2.
\]
As discussed after \eqref{eq:one-round}, using the preceding centered
utility $\vec u_i^{(t-1)}$ in place of the unavailable $\ut_i$
introduces the squared prediction-error term
$\norm{\ut_i-\vec u_i^{(t-1)}}_\infty^2$. We therefore seek a movement
bound involving these same two quantities.

This requirement guides the potential construction. In \algname{}, the
positive coordinates $\partial_a\Psi(\eta\vec U^{(t)})$ serve as the
base weights assigned to the actions before normalization. Since
square-root movement depends on relative changes in these weights, we
need their ratios across consecutive rounds to remain controlled even
when the weights themselves become very small. In other words, these
coordinates must be multiplicatively stable. The next subsection
constructs $\Psi$ to ensure this property.

\subsection{Choosing the Potential}\label{sec:potential}

We now construct the potential $\Psi$ used by \algname{}. The preceding
discussion identifies three properties that the potential must satisfy.
First, $\Psi$ should dominate the largest cumulative advantage
$\max_a\vec U[a]$ while having initial value only $\mathcal{O}(\log d)$, so that
a bound on the potential yields a regret bound with logarithmic
dependence on the number of actions. Second, the curvature of $\Psi$
should be controlled by its coordinates
$\partial_a\Psi(\vec U)$, as required for the weighted Taylor estimate
\eqref{eq:taylor}. Third, these same coordinates must remain sufficiently
stable under normalization (multiplicatively stable) so that the negative weighted square in
\eqref{eq:one-round} can control the squared Hellinger movement from
\Cref{sec:hellinger}. We construct $\Psi$ step by step to satisfy these
three requirements.

\subsubsection{The Scalar Function}
\label{sec:potential-scalar}

We begin by designing a one-dimensional function $f$ that will underlie
the potential and its action weights. To convert a potential bound into
a regret bound, we need
\[
    f(z)\geq1+z.
\]
On the nonnegative half-line, the simplest choice is therefore
\[
    f(z)=1+z,
    \qquad z\geq0.
\]
This gives the desired linear control of positive cumulative advantages
and has constant derivative $f'(z)=1$.

The negative half-line requires more care. Extending the linear branch
as $[1+z]_+$ would make the derivative vanish for $z<-1$, so actions
with sufficiently negative cumulative advantage would receive zero
weight. More importantly, we need the resulting positive weights to
remain stable after normalization. As discussed in
\Cref{sec:hellinger}, the negative weighted square in
\eqref{eq:one-round} must ultimately control squared Hellinger movement,
which in turn controls the prediction error in a fixed game.

Absolute stability is not enough after normalization. For example, the
positive vectors $(\eps,\eps)$ and $(2\eps,\eps)$ become arbitrarily
close as $\eps\to0$, while their normalizations remain
$(1/2,1/2)$ and $(2/3,1/3)$. Thus, when weights become small, their
relative changes must become small as well. This type of multiplicative
stability also appears in the analysis of log-regularized FTRL and
Cautious Optimism
\citep{farina2022near,soleymani2025faster,soleymani2025cautious}.

We therefore seek a positive, convex, continuously differentiable
function $f$ that agrees with $1+z$ for $z\geq0$, satisfies $f'(z)>0$
at every finite input, approaches zero as $z\to-\infty$, and becomes
increasingly stable in relative terms as $f'(z)$ becomes small.

To identify a sufficient condition for the last property, consider a
scalar prototype. For an input vector $\vec z$, normalize the derivative
weights according to
\[
    \vx[a]
      =
      \frac{f'(\vec z[a])}
           {\sum_b f'(\vec z[b])}.
\]
If the input changes by a small vector $\vec h$, then the relative
change of the weight of action $a$ is, to first order,
\[
    \frac{f'(\vec z[a]+\vec h[a])-f'(\vec z[a])}
         {f'(\vec z[a])}
    \approx
    \frac{f''(\vec z[a])}{f'(\vec z[a])}\vec h[a].
\]
Hence, $f''(z)/f'(z)$ measures the local relative sensitivity of the
derivative weight.

Let $\vx'$ denote the normalized weights after the perturbation and set
\[
    S=\sum_b f'(\vec z[b]).
\]
The local normalization calculation underlying squared Hellinger
movement gives
\[
    \norm{\sqrt{\vx'}-\sqrt{\vx}}_2^2
    \lesssim
    \frac1S
    \sum_a
       f'(\vec z[a])
       \left(
          \frac{f''(\vec z[a])}{f'(\vec z[a])}
       \right)^2
       \vec h[a]^2.
\]
The important feature is the factor $1/S$ created by normalization.
We want to control this expression by $
    \sum_a f'(\vec z[a])\vec h[a]^2$, which is the scalar analogue of the gradient-weighted square retained
in \eqref{eq:one-round}.

A convenient sufficient condition is
\[
    \left(\frac{f''(z)}{f'(z)}\right)^2
      \lesssim f'(z).
    \numberthis{eq:suff_cond_f}
\]
Indeed, under \eqref{eq:suff_cond_f},
\[
    \frac1S
    \sum_a
       f'(\vec z[a])
       \left(
          \frac{f''(\vec z[a])}{f'(\vec z[a])}
       \right)^2
       \vec h[a]^2
    \lesssim
    \frac1S
    \sum_a f'(\vec z[a])^2\vec h[a]^2
    \leq
    \sum_a f'(\vec z[a])\vec h[a]^2,
\]
where the last inequality uses $f'(\vec z[a])\leq S$. Thus, the extra
factor of $f'(z)$ in \eqref{eq:suff_cond_f} compensates for the
normalizing denominator. This calculation only motivates the scalar
design. After aggregation, \eqref{eq:compensation} provides the
corresponding bound for the actual gradient weights of $\Psi$.

We now choose the negative branch to satisfy \eqref{eq:suff_cond_f}.
A particularly simple choice is
\[
    f'(z)=f(z)^2,
    \qquad
    f(0)=1,
    \qquad
    z\leq0.
\]
Indeed, $f''(z)=2 f'(z) f(z) = 2f(z)^3$, so
\[
    \left(\frac{f''(z)}{f'(z)}\right)^2
      =4f(z)^2
      =4f'(z).
\]
Thus, the squared relative sensitivity decreases at exactly the same
rate as the derivative weight.

The differential equation satisfies
\[
    \frac{d}{dz}\frac1{f(z)}
      =-1,
\]
and the condition $f(0)=1$ therefore gives
\[
    f(z)=\frac1{1-z},
    \qquad z\leq0.
\]
Joining the reciprocal negative branch with the linear positive branch
yields
\[
    f(z)
      =
      \begin{cases}
         (1-z)^{-1}, & z\leq0,\\
         1+z, & z\geq0.
      \end{cases}
\]
The two branches agree in value and first derivative at zero, so $f$ is
positive, convex, and continuously differentiable, with
$f'(z)=\min\{f(z)^2,1\}$.

Finally, the reciprocal branch preserves the lower bound needed for
regret. For $z\leq0$,
\[
    f(z)-(1+z)
      =
      \frac{z^2}{1-z}
      \geq0,
\]
while equality holds for $z\geq0$. Hence, $f(z)\geq1+z$ everywhere.

The reciprocal tail therefore provides both ingredients we need from the
scalar construction. It preserves the linear lower bound used to control
regret, while making small derivative weights sufficiently stable under
normalization. The next step aggregates these scalar values across
actions while retaining these properties and keeping the initial
potential only logarithmic in $d$.

\subsubsection{Aggregation and Scaling}
\label{sec:potential-aggregation}

We now aggregate the scalar values $f(\vec U[a]/c)$ across actions.
The aggregation must preserve the lower bound on the largest cumulative
advantage while keeping the initial potential small. A direct sum would
not achieve the latter, since $f(0)=1$ gives
\[
    \sum_{a\in\mathcal A_i} f(0)
      =
      |\mathcal A_i|,
\]
which can be of order $d$.

Instead, we use an $\ell_{c-1}$ norm. This follows the classical
$p$-norm idea of choosing the norm order logarithmic in the dimension,
so that a high-order norm approximates the maximum while having only
logarithmic dimension dependence
\citep{gentile2003robustness,shalev2012online}. For any nonnegative vector,
\[
    \max_a f\!\left(\frac{\vec U[a]}c\right)
    \leq
    \left(
       \sum_{a\in\mathcal A_i}
       f\!\left(\frac{\vec U[a]}c\right)^{c-1}
    \right)^{1/(c-1)}
    \leq
    |\mathcal A_i|^{1/(c-1)}
    \max_a f\!\left(\frac{\vec U[a]}c\right).
\]
With $c=2+\log d$, we have $c-1=1+\log d$, and therefore
\[
    |\mathcal A_i|^{1/(c-1)}
      \leq
      d^{1/(1+\log d)}
      <e<3.
\]
Thus, the $\ell_{c-1}$ norm approximates the largest scalar value within
a constant factor. In particular, at the origin its value is less than
three. This is the reason for choosing an aggregation exponent of order
$\log d$.

We also divide the scalar inputs by $c$ and multiply the resulting norm
by $c$. These two scalings are chosen so that differentiation does not
introduce an additional factor depending on $c$. Indeed, the factor
$1/c$ from differentiating $f(\vec U[a]/c)$ is canceled by the outer
multiplier $c$, while the factors $c-1$ from the inner and outer powers
cancel each other. This leads to
\[
    \Psi_i(\vec U)
       =
       c\left(
          \sum_{a\in\mathcal A_i}
          f\!\left(\frac{\vec U[a]}c\right)^{c-1}
       \right)^{1/(c-1)},
\]
which is the potential defined in \eqref{eq:potential}.

The construction now gives the two properties needed for the regret
analysis. First, since $f(0)=1$,
\[
    \Psi_i(\zero)
      =
      c|\mathcal A_i|^{1/(c-1)}
      <3c.
\]
Second, the $\ell_{c-1}$ norm dominates each coordinate and
$f(z)\geq1+z$, so for every action $a$,
\[
    \Psi_i(\vec U)
      \geq
      c f\!\left(\frac{\vec U[a]}c\right)
      \geq
      c+\vec U[a].
\]
Hence,
\[
    0<\Psi_i(\vec U),
    \qquad
    \Psi_i(\zero)<3c,
    \qquad
    \Psi_i(\vec U)
      \geq
      c+\max_a\vec U[a].
    \numberthis{eq:certificate}
\]

The last inequality is the regret certificate. If
$\Psi_i(\eta\vec U_i^{(T+1)})\leq4c$, then
\[
    \eta\Reg_i^{(T)}
      =
      \eta\max_a\vec U_i^{(T+1)}[a]
      \leq
      \Psi_i(\eta\vec U_i^{(T+1)})-c
      \leq
      3c,
\]
where we used \eqref{eq:two-identities}. Therefore,
\[
    \Reg_i^{(T)}
      \leq
      \frac{3c}{\eta}.
\]
Since $c=2+\log d$, the initial potential is only $O(\log d)$.
Thus, the power-norm aggregation is what ultimately gives the logarithmic
dependence on the number of actions.

\subsubsection{Recovering the Power Weights}
\label{sec:potential-weights}

The potential construction also explains the explicit power weights used
by \algname{}. Differentiating $\Psi_i$ with respect to coordinate $a$
gives
\[
    \partial_a\Psi_i(\vec U)
      =
      \left(\frac{\Psi_i(\vec U)}c\right)^{2-c}
      f\!\left(\frac{\vec U[a]}c\right)^{c-2}
      f'\!\left(\frac{\vec U[a]}c\right).
\]
The first factor is common to all actions. The action-specific part is
determined entirely by the two branches of $f$.

On the negative branch, $f'=f^2$, while on the positive branch,
$f'=1$. Therefore,
\[
    f\!\left(\frac{\vec U[a]}c\right)^{c-2}
    f'\!\left(\frac{\vec U[a]}c\right)
    =
    \begin{cases}
       \left(1-\frac{\vec U[a]}c\right)^{-c},
          & \vec U[a]\leq0,\\
       \left(1+\frac{\vec U[a]}c\right)^{c-2},
          & \vec U[a]\geq0.
    \end{cases}
\]
This is exactly the weight function $a_c$ defined in
\eqref{eq:weight}. Hence, we recover
\[
    \partial_a\Psi_i(\vec U)
      =
      \left(\frac{\Psi_i(\vec U)}c\right)^{2-c}
      a_c(\vec U[a]),
\]
which is \eqref{eq:gradient}.

The prefactor $(\Psi_i(\vec U)/c)^{2-c}$ is independent of the action,
so it disappears when the weights are normalized. Substituting
\eqref{eq:gradient} into the gradient response
\eqref{eq:gradient-rule} therefore gives
\[
    \xt_i[a]
      \propto
      a_c\!\left(\eta\vec U_i^{(t)}[a]\right)
      \bigl(1+4\eta\vec u_i^{(t-1)}[a]\bigr),
\]
which is precisely the explicit update in \eqref{eq:explicit-rule}.

Thus, the two powers appearing in $a_c$ arise directly from the scalar
construction. The reciprocal branch $f'=f^2$ produces the exponent
$c$ for negative cumulative advantages, while the linear branch
$f'=1$ produces the exponent $c-2$ for positive cumulative advantages.

\subsubsection{Weighted Curvature and Stability}
\label{sec:potential-curvature}

We now verify the geometric properties of $\Psi_i$ needed in the
analysis. The first is a weighted curvature bound, which gives the
Taylor estimate \eqref{eq:taylor}. The second is multiplicative
stability of the coordinates $\partial_a\Psi_i$, which allows us to
compare these weights along a finite update. We also record a uniform
bound on their total mass, which will later be used to control
prediction-error terms.

Recall the scalar weight $a_c$ from \eqref{eq:weight}. Away from zero,
its logarithmic derivative is
\[
    \frac{a_c'(z)}{a_c(z)}
      =
      \begin{cases}
         (1-z/c)^{-1}, & z<0,\\
         (c-2)/(c+z), & z>0.
      \end{cases}
    \numberthis{eq:weight-log-derivative}
\]
Both branches lie in $[0,1]$, independently of $d$. Thus, the relative
sensitivity of each scalar weight is uniformly bounded.

This bound also controls the curvature of the aggregated potential.
Differentiating \eqref{eq:gradient} at points with no zero coordinates
gives
\[
    \nabla^2\Psi_i(\vec U)
      =
      \diag\!\left(
         \partial_a\Psi_i(\vec U)
         \frac{a_c'(\vec U[a])}{a_c(\vec U[a])}
      \right)
      -
      \frac{c-2}{\Psi_i(\vec U)}
      \nabla\Psi_i(\vec U)\nabla\Psi_i(\vec U)^\top.
\]
The rank-one matrix in the second term is positive semidefinite, so
subtracting it can only decrease the Hessian in the positive
semidefinite order. Using \eqref{eq:weight-log-derivative}, the diagonal
term is bounded above by
$\diag(\partial_a\Psi_i(\vec U))$. Since $\Psi_i$ is convex, its Hessian
is positive semidefinite wherever it exists. Hence,
\[
    0
      \preceq
      \nabla^2\Psi_i(\vec U)
      \preceq
      \diag\bigl(\partial_a\Psi_i(\vec U)\bigr).
\]
Equivalently, for every vector $\vec v$,
\[
    \vec v^\top\nabla^2\Psi_i(\vec U)\vec v
      \leq
      \sum_a
      \partial_a\Psi_i(\vec U)\vec v[a]^2.
\]
This is exactly the gradient-weighted curvature property motivated in
\Cref{sec:optimism}.

We will also use that the gradient has uniformly bounded total mass.
Applying H\"older's inequality to the derivative formula gives
\[
    \sum_a\partial_a\Psi_i(\vec U)\leq3.
\]
We record the two bounds together as
\[
    \sum_a\partial_a\Psi_i(\vec U)\leq3,
    \qquad
    0\preceq\nabla^2\Psi_i(\vec U)
      \preceq
      \diag\bigl(\partial_a\Psi_i(\vec U)\bigr)
    \quad\text{where the Hessian exists}.
    \numberthis{eq:mass-curvature}
\]
The Hessian bound controls the Taylor remainder, while the total-mass
bound will later convert gradient-weighted prediction errors into
$\ell_\infty$ prediction errors.

The remaining ingredient is multiplicative stability. By
\Cref{lem:multiplicative}, moving the potential input from
$\eta\vec U_i^{(t)}$ in the direction $\eta\ut_i$ changes each gradient
coordinate only multiplicatively. In particular, for
$0\leq\theta\leq1$,
\[
    \partial_a\Psi_i
       \bigl(\eta\vec U_i^{(t)}+\theta\eta\ut_i\bigr)
    \leq
    e^{2\theta\eta\norm{\ut_i}_\infty}
    \partial_a\Psi_i(\eta\vec U_i^{(t)}).
\]
Since $\norm{\ut_i}_\infty\leq1$ and $\eta\leq1/8$,
\[
    e^{2\theta\eta}
      \leq
      e^{2\eta}
      <
      \frac43.
\]
Thus, every gradient coordinate along the update segment is at most
$4/3$ times its value at the beginning of the round.

Combining this multiplicative stability with
\eqref{eq:mass-curvature} gives the desired Taylor estimate. Indeed, the
integral remainder is at most
\[
    \eta^2
    \int_0^1
       (1-\theta)
       \sum_a
       \partial_a\Psi_i
          \bigl(\eta\vec U_i^{(t)}+\theta\eta\ut_i\bigr)
       \ut_i[a]^2
    \,d\theta,
\]
which is bounded by
\[
    \frac{4\eta^2}{3}
    \int_0^1(1-\theta)\,d\theta
    \sum_a
       \partial_a\Psi_i(\eta\vec U_i^{(t)})
       \ut_i[a]^2
    =
    \frac{2\eta^2}{3}
    \sum_a
       \partial_a\Psi_i(\eta\vec U_i^{(t)})
       \ut_i[a]^2.
\]
Adding the first-order term gives \eqref{eq:taylor}.
\Cref{app:calculus} proves these properties formally, including the
regularity needed when a coordinate crosses zero.

\subsubsection{Stability After Normalization}
\label{sec:normalization}

We now lift the scalar stability condition from
\Cref{sec:potential-scalar} to the gradient weights of the aggregated
potential. The main difficulty is normalization. When the total
derivative mass
$\sum_a\partial_a\Psi_i(\vec U)$ becomes small, controlling only the
absolute or multiplicative change of each weight is not enough. We need
the sensitivity of the weights to decrease at the same time.

The role of the reciprocal negative branch is especially transparent
when every coordinate of the input equals $cz$ with $z<0$. In this
case,
\[
    \sum_a\partial_a\Psi_i(cz\one)
      =
      |\mathcal A_i|^{1/(c-1)}f(z)^2,
    \qquad
    \left(
       \frac{a_c'(cz)}{a_c(cz)}
    \right)^2
      =
      f(z)^2.
\]
Thus, the total derivative mass and the squared logarithmic sensitivity
vanish at the same rate. For arbitrary inputs, the corresponding bound
is
\[
    \frac{
       \bigl(a_c'(\vec U[a])/a_c(\vec U[a])\bigr)^2
    }{
       \sum_b\partial_b\Psi_i(\vec U)
    }
    \leq3
    \quad\text{where the derivatives exist}.
    \numberthis{eq:compensation}
\]
This is the aggregated analogue of the scalar condition developed in
\Cref{sec:potential-scalar}. Importantly, it does not require the total
derivative mass to be bounded away from zero. Instead, the numerator
shrinks whenever the denominator does. \Cref{app:movement} proves
\eqref{eq:compensation}.

To see why this is exactly the normalization bound we need, first ignore
the optimistic correction. By \eqref{eq:gradient}, normalizing the
power weights $a_c(\vec U[a])$ is equivalent to normalizing the gradient
coordinates,
\[
    \vx[a]
      =
      \frac{a_c(\vec U[a])}{\sum_b a_c(\vec U[b])}
      =
      \frac{\partial_a\Psi_i(\vec U)}
           {\sum_b\partial_b\Psi_i(\vec U)}.
\]
For a small input change in direction $\vec v$, the square-root
movement is controlled locally by
\[
    \frac14
    \sum_a
       \frac{\partial_a\Psi_i(\vec U)}
            {\sum_b\partial_b\Psi_i(\vec U)}
       \left(
          \frac{a_c'(\vec U[a])}{a_c(\vec U[a])}
       \right)^2
       \vec v[a]^2.
\]
Applying \eqref{eq:compensation} removes the normalizing denominator and
gives
\[
    \frac34
    \sum_a
       \partial_a\Psi_i(\vec U)\vec v[a]^2.
\]
Thus, after normalization, strategy movement is controlled by the same
gradient-weighted square that appears in \eqref{eq:one-round}.

Along the learning trajectory, the cumulative input changes in direction
$\eta\ut_i$. Integrating the preceding local estimate and using the
multiplicative stability established in
\Cref{sec:potential-curvature} gives a finite-step bound in terms of
\[
    \eta^2
    \sum_a
       \partial_a\Psi_i(\eta\vec U_i^{(t)})\ut_i[a]^2.
\]
The optimistic correction also changes between rounds, from
$\one+4\eta\vec u_i^{(t-1)}$ to $\one+4\eta\ut_i$. As discussed after
\eqref{eq:one-round}, this second change is controlled by the squared
prediction error
$\norm{\ut_i-\vec u_i^{(t-1)}}_\infty^2$.

Consequently, the movement between consecutive strategies is controlled
by exactly the two quantities already present in the optimistic
potential inequality. \Cref{lem:movement} makes this statement precise.

\subsection{Bounding Regret in Self-play}

We now combine the potential inequality with the strategy-movement and
prediction-error bounds developed above. The key point is that the
negative gradient-weighted square in \eqref{eq:one-round} controls the
same squared Hellinger movement that, in a fixed game, controls the
prediction error.

By \Cref{lem:movement}, consecutive strategies satisfy
\[
    \norm{\sqrt{\vx_i^{(t+1)}}-\sqrt{\xt_i}}_2^2
    \leq
    4\eta^2\sum_a
       \partial_a\Psi_i(\eta\vec U_i^{(t)})\ut_i[a]^2
    +
    16\eta^2
       \norm{\ut_i-\vec u_i^{(t-1)}}_\infty^2.
\]
The first term is exactly the gradient-weighted square that appears with
a negative sign in \eqref{eq:one-round}, while the second is the same
squared prediction error that appears on its positive side.

Summing the movement estimate over rounds and combining it with the
telescoped version of \eqref{eq:one-round} gives
\[
    \Psi_i(\eta\vec U_i^{(T+1)})
    \leq
    3c
    +
    10\eta^2\sum_{t=1}^T
       \norm{\ut_i-\vec u_i^{(t-1)}}_\infty^2 -  \frac14\sum_{t=2}^T
       \norm{\sqrt{\xt_i}-\sqrt{\vx_i^{(t-1)}}}_2^2.
    \numberthis{eq:rvu}
\]
We refer to \eqref{eq:rvu} as a \emph{potential-level RVU inequality}. It is
derived directly from the potential argument and holds for arbitrary
bounded utility vectors. The left-hand side contains both the terminal
potential and accumulated squared Hellinger movement, while the
right-hand side contains the initial-potential bound and accumulated
squared prediction error. No self-play in games assumption has been used yet.

We now use self-play in a fixed game. By \eqref{eq:game-error},
for every $t\geq2$,
\[
    \norm{\ut_i-\vec u_i^{(t-1)}}_\infty^2
    \leq
    5\sum_{j=1}^n
       \norm{\sqrt{\xt_j}-\sqrt{\vx_j^{(t-1)}}}_2^2.
\]
Thus, the prediction error on the right-hand side of \eqref{eq:rvu} is
controlled by exactly the same movement quantity that appears on its
left-hand side. For the first round,
$\vec u_i^{(0)}=\zero$ and
$\norm{\vec u_i^{(1)}}_\infty\leq1$.

Substituting \eqref{eq:game-error} into \eqref{eq:rvu}, summing over
players, and discarding the nonnegative terminal potentials gives
\[
    \left(\frac14-50n\eta^2\right)
    \sum_{j=1}^n\sum_{t=2}^T
       \norm{\sqrt{\xt_j}-\sqrt{\vx_j^{(t-1)}}}_2^2
    \leq
    3nc+10n\eta^2.
\]
Only one factor of $n$ appears in the movement coefficient because
summing the prediction-error bounds contributes one factor of $n$.
The Hellinger product comparison in \Cref{sec:hellinger} avoids the
additional player factor that arises in the corresponding $\ell_1$
argument.

With $\eta=1/(32\sqrt n)$,
\[
    n\eta^2=\frac1{1024},
    \qquad
    \frac14-50n\eta^2
      =\frac{103}{512}
      >\frac15.
\]
The movement term can therefore be absorbed onto the left-hand side,
yielding
\[
    \sum_{j=1}^n\sum_{t=2}^T
       \norm{\sqrt{\xt_j}-\sqrt{\vx_j^{(t-1)}}}_2^2
    \leq
    16nc.
\]
In particular, the accumulated squared movement is bounded uniformly
over the horizon.

We can now return to the potential of a single player. Substituting the
movement bound into \eqref{eq:rvu} and using
\eqref{eq:game-error} gives
\[
    \Psi_i(\eta\vec U_i^{(T+1)})
    \leq
    4c.
\]
Finally, the certificate \eqref{eq:certificate} and the regret identity
\eqref{eq:two-identities} imply
\[
    \eta\Reg_i^{(T)}
      \leq
      \Psi_i(\eta\vec U_i^{(T+1)})-c
      \leq
      3c.
\]
Hence,
\[
    \Reg_i^{(T)}
      \leq
      \frac{3c}{\eta}
      =
      96\sqrt n\,(2+\log d).
\]

The dependence on $d$ comes from the initial potential
$c=2+\log d$, while the $\sqrt n$ dependence comes from choosing
$\eta$ so that $n\eta^2$ remains a sufficiently small constant. The
argument applies to every finite horizon and bounds accumulated squared
strategy movement.

\section{Analysis of \algname{}}\label{sec:analysis}

We now make the proof overview rigorous. We first establish the
geometric properties of the potential needed for the weighted Taylor
bound and for stability after normalization. These properties yield a
one-round potential inequality with a negative gradient-weighted square.
We then show that the same weighted square controls strategy movement.
Combining the resulting potential and movement bounds gives a
potential-level inequality in terms of the squared prediction errors.
Finally, in a fixed game, we bound these prediction errors by the
players' squared Hellinger movement and close the argument to obtain a
uniform bound on the potential and regret.

Until the self-play analysis in \Cref{sec:analysis-game}, we fix a player $i$ and allow the observed
utility vectors $\nut_i$ to be an arbitrary sequence in
$[0,1]^{|\mathcal A_i|}$. Throughout this part, we set $c=2+\log d$
and use a fixed learning rate $0<\eta\leq1/16$. We specialize to
self-play in a fixed game and set $\eta=1/(32\sqrt n)$ only in the
final part of the analysis.

Recall from \eqref{eq:centered} and \eqref{eq:cumulative} that
\[
    \ut_i
      =
      \nut_i-\ip{\xt_i}{\nut_i}\one,
    \qquad
    \vec U_i^{(t+1)}
      =
      \vec U_i^{(t)}+\ut_i.
\]
Hence,
\[
    \norm{\ut_i}_\infty\leq1,
    \qquad
    \ip{\xt_i}{\ut_i}=0,
    \qquad
    \vec U_i^{(1)}=\vec u_i^{(0)}=\zero.
\]
All sums over actions in a single-player calculation are over
$\mathcal A_i$, and every horizon considered below is finite.

The basic calculus properties of the potential, including its weighted
curvature and multiplicative stability, are proved in
\Cref{app:calculus}. The corresponding stability bounds for the
normalized response are proved in \Cref{app:movement}. The elementary
square-root-distance identities used to compare expectations and product
distributions (properties of Hellinger distance) are collected in \Cref{app:game}. We state the results we
need below before applying them.

The adversarial guarantee is treated separately in
\Cref{app:safeguard}. The safeguard leaves the \algname{} update
unchanged and only decreases the learning rate when necessary under
adversarial utilities.

\subsection{Potential Geometry and Multiplicative Stability}
\label{sec:analysis-geometry}

We now make precise the potential properties motivated in the proof
overview in \Cref{sec:potential}. In \Cref{sec:optimism}, we required the curvature of $\Psi_i$
to be controlled by its partial derivatives
$\partial_a\Psi_i$, so that the Taylor remainder carries the same
weights as the optimistic potential decrease. In
\Cref{sec:normalization}, we saw that these weights must also remain
stable under finite changes of the potential input. We establish these
properties below and then combine them to obtain the finite-step Taylor
bound.

We begin with the basic geometry of the potential, collecting its
convexity, positivity of the derivatives, and weighted curvature bound.

\begin{lemma}
\label{lem:calculus}
Fix a player $i$ and let $c=2+\log d$. The potential $\Psi_i$ is
convex and continuously differentiable on $\R^{|\mathcal A_i|}$, with
locally Lipschitz gradient. For every $\vec U$,
\[
    \Psi_i(\vec U)>0,
    \qquad
    \Psi_i(\zero)
      =c|\mathcal A_i|^{1/(c-1)}
      <3c,
    \qquad
    \Psi_i(\vec U)
      \geq c+\max_a\vec U[a].
\]
Moreover,
\[
    \partial_a\Psi_i(\vec U)
      =
      \left(\frac{\Psi_i(\vec U)}c\right)^{2-c}
      a_c(\vec U[a])
      >0,
    \qquad
    \sum_a\partial_a\Psi_i(\vec U)\leq3.
\]
Where the Hessian exists,
\[
    0
      \preceq
      \nabla^2\Psi_i(\vec U)
      \preceq
      \diag\bigl(\partial_a\Psi_i(\vec U)\bigr).
\]
The same directional second-derivative bound holds almost everywhere
along every line segment.
\end{lemma}

The bounds
$\Psi_i(\zero)<3c$ and
$\Psi_i(\vec U)\geq c+\max_a\vec U[a]$ give the regret certificate
from \eqref{eq:certificate}. In particular, controlling the terminal
potential will directly control the largest cumulative regret.

The remaining bounds describe the geometry needed for the potential
analysis. The total derivative mass satisfies
$\sum_a\partial_a\Psi_i(\vec U)\leq3$, which will allow us to bound
gradient-weighted prediction errors by their $\ell_\infty$ counterparts.
More importantly, the Hessian bound gives, for every vector $\vec v$,
\[
    \vec v^\top\nabla^2\Psi_i(\vec U)\vec v
      \leq
      \sum_a
         \partial_a\Psi_i(\vec U)\vec v[a]^2.
\]
Thus, the curvature of $\Psi_i$ is controlled by the same coordinates
$\partial_a\Psi_i(\vec U)$ that determine the action weights in
\algname{}. This is precisely the weighted-curvature property motivated
in \Cref{sec:optimism}. \Cref{app:calculus} proves
\Cref{lem:calculus} directly from the scalar function $f$ and the
power-norm construction of $\Psi_i$.

The Hessian bound is pointwise, whereas a finite update moves through
an entire segment of potential inputs. To control the Taylor remainder
using the gradient at the beginning of the update, we therefore need to
compare $\partial_a\Psi_i$ at different points along this segment. This
is the role of multiplicative stability.

\begin{lemma}[Multiplicative stability]\label{lem:multiplicative}
For every $\vec U,\vec u\in\R^{|\mathcal A_i|}$, every $\eta>0$,
and every action $a\in\mathcal A_i$,
\[
    \left|
       \log
       \frac{\partial_a\Psi_i(\vec U+\eta\vec u)}
            {\partial_a\Psi_i(\vec U)}
    \right|
    \leq
    2\eta\norm{\vec u}_\infty.
    \numberthis{eq:gradient-multiplicative}
\]
For \algname{} with a fixed rate $0<\eta\leq1/16$ and any sequence of
utility vectors in $[0,1]^{|\mathcal A_i|}$ observed after play,
\[
    \left|
       \log
       \frac{\vx_i^{(t+1)}[a]}{\xt_i[a]}
    \right|
    \leq
    2\eta\norm{\ut_i}_\infty
    +
    \frac{8\eta}{1-4\eta}
       \norm{\ut_i-\vec u_i^{(t-1)}}_\infty.
    \numberthis{eq:iterate-multiplicative}
\]
In particular, if $\eta\leq1/32$, then
$\vx_i^{(t+1)}[a]/\xt_i[a]\in[1/2,2]$ for every action $a$.
\end{lemma}

The first inequality gives multiplicative stability of the potential
weights $\partial_a\Psi_i$. The second shows that the played
probabilities inherit a corresponding stability bound, accounting for
both the cumulative-utility update and the optimistic correction. For
the Taylor argument, we only need
\eqref{eq:gradient-multiplicative}; the probability-ratio bound will be
used later in the movement analysis. Both statements are proved in
\Cref{app:calculus}.

Combining the pointwise curvature bound from \Cref{lem:calculus} with
the multiplicative stability of the gradient coordinates from \Cref{lem:multiplicative} gives the
finite-step Taylor estimate used throughout the analysis.

\begin{lemma}[One-step Taylor bound]\label{lem:taylor}
For every $\vec U,\vec u\in\R^{|\mathcal A_i|}$ with
$\norm{\vec u}_\infty\leq1$ and every $0<\eta\leq1/8$,
\[
    \Psi_i(\vec U+\eta\vec u)-\Psi_i(\vec U)
    \leq
    \eta\ip{\nabla\Psi_i(\vec U)}{\vec u}
    +
    \frac{2\eta^2}{3}
    \sum_a
       \partial_a\Psi_i(\vec U)\vec u[a]^2.
\]
\end{lemma}

The coefficient $2/3$ comes from combining the weighted curvature bound
with multiplicative stability along the update segment. By
\eqref{eq:gradient-multiplicative},
\[
    \partial_a\Psi_i(\vec U+\theta\eta\vec u)
      \leq
      \frac43\partial_a\Psi_i(\vec U),
    \qquad 0\leq\theta\leq1,
\]
when $\norm{\vec u}_\infty\leq1$ and $\eta\leq1/8$. Hence, the
integral Taylor remainder is at most
\[
    \frac{4\eta^2}{3}
    \int_0^1(1-\theta)\,d\theta
    \sum_a
       \partial_a\Psi_i(\vec U)\vec u[a]^2
    =
    \frac{2\eta^2}{3}
    \sum_a
       \partial_a\Psi_i(\vec U)\vec u[a]^2.
\]
Adding the first-order term proves \Cref{lem:taylor}. Taking
$\vec U=\eta\vec U_i^{(t)}$ and $\vec u=\ut_i$ recovers
\eqref{eq:taylor}. The full argument, including the regularity needed
when a coordinate crosses zero, is given in \Cref{app:calculus}.

\subsection{One-round and Cumulative Potential Bounds}
\label{sec:analysis-potential}

We now make the cancellation from \Cref{sec:optimism} rigorous. We first
derive a one-round potential inequality. This argument is local to a
single round and therefore remains valid when the safeguard changes the
learning rate between rounds. We then specialize to a fixed learning
rate and telescope the one-round inequalities to obtain the cumulative
potential bound used in the self-play analysis.

\begin{lemma}[One-round potential change]\label{lem:one-round}
Fix a player $i$ and a round $t$. Suppose that on this round the player
uses the \algname{} response \eqref{eq:gradient-rule} with some
$0<\eta\leq1/16$. After observing
$\nut_i\in[0,1]^{|\mathcal A_i|}$ and updating according to
\eqref{eq:centered} and \eqref{eq:cumulative},
\[
    \Psi_i(\eta\vec U_i^{(t+1)})
      -
    \Psi_i(\eta\vec U_i^{(t)})
    \leq
    2\eta^2
    \sum_a
       \partial_a\Psi_i(\eta\vec U_i^{(t)})
       \bigl(\ut_i[a]-\vec u_i^{(t-1)}[a]\bigr)^2
    -
    \eta^2
    \sum_a
       \partial_a\Psi_i(\eta\vec U_i^{(t)})
       \ut_i[a]^2.
\]
The bound is local to round $t$ and remains valid even if the learning rate changes across rounds.
\end{lemma}

\begin{proof}
Since $\norm{\vec u_i^{(t-1)}}_\infty\leq1$ and
$\eta\leq1/16$, every optimistic correction factor is at least $3/4$.
Together with the positivity of
$\partial_a\Psi_i$ from \Cref{lem:calculus}, this ensures that the
response is well defined.

By the centering identity in \eqref{eq:two-identities},
$\ip{\xt_i}{\ut_i}=0$. Substituting the normalized response and
clearing its positive normalizing denominator gives
\[
    0
      =
      \sum_a
         \partial_a\Psi_i(\eta\vec U_i^{(t)})\ut_i[a]
      +
      4\eta
      \sum_a
         \partial_a\Psi_i(\eta\vec U_i^{(t)})
         \vec u_i^{(t-1)}[a]\ut_i[a].
\]
Hence, the first-order term in \Cref{lem:taylor} satisfies
\[
    \eta
    \ip{\nabla\Psi_i(\eta\vec U_i^{(t)})}{\ut_i}
      =
      -4\eta^2
      \sum_a
         \partial_a\Psi_i(\eta\vec U_i^{(t)})
         \vec u_i^{(t-1)}[a]\ut_i[a].
\]
For every action $a$,
\[
    -4\vec u_i^{(t-1)}[a]\ut_i[a]
      =
      2\bigl(\ut_i[a]-\vec u_i^{(t-1)}[a]\bigr)^2
      -
      2\ut_i[a]^2
      -
      2\vec u_i^{(t-1)}[a]^2.
\]
Dropping the last, nonpositive contribution and applying
\Cref{lem:taylor} with input $\eta\vec U_i^{(t)}$ and direction
$\ut_i$ gives
\[
    \Psi_i(\eta\vec U_i^{(t+1)})
      -
    \Psi_i(\eta\vec U_i^{(t)})
    \leq
    2\eta^2
    \sum_a
       \partial_a\Psi_i(\eta\vec U_i^{(t)})
       \bigl(\ut_i[a]-\vec u_i^{(t-1)}[a]\bigr)^2
    -
    \frac{4\eta^2}{3}
    \sum_a
       \partial_a\Psi_i(\eta\vec U_i^{(t)})
       \ut_i[a]^2.
\]
Only the learning rate used on round $t$ enters the argument.
\end{proof}

This is the rigorous version of the one-round estimate
\eqref{eq:one-round} from the proof overview. We now fix the learning
rate so that the potential differences telescope across rounds.

\begin{lemma}[Potential telescope bound]\label{lem:budget}
For any fixed $0<\eta\leq1/16$, any sequence of utility vectors in
$[0,1]^{|\mathcal A_i|}$ observed after play, and every $T\geq1$,
\[
    \Psi_i(\eta\vec U_i^{(T+1)})
    +
    \eta^2
    \sum_{t=1}^T\sum_a
       \partial_a\Psi_i(\eta\vec U_i^{(t)})\ut_i[a]^2
    \leq
    3c
    +
    6\eta^2
    \sum_{t=1}^T
       \norm{\ut_i-\vec u_i^{(t-1)}}_\infty^2.
    \numberthis{eq:budget}
\]
\end{lemma}

\begin{proof}
Apply \Cref{lem:one-round} on each round and move the negative weighted
square to the left. By the total derivative-mass bound in
\Cref{lem:calculus},
\[
    \sum_a
       \partial_a\Psi_i(\eta\vec U_i^{(t)})
       \bigl(\ut_i[a]-\vec u_i^{(t-1)}[a]\bigr)^2
    \leq
    3\norm{\ut_i-\vec u_i^{(t-1)}}_\infty^2.
\]
Therefore,
\[
    \Psi_i(\eta\vec U_i^{(t+1)})
      -
    \Psi_i(\eta\vec U_i^{(t)})
    +
    \eta^2
    \sum_a
       \partial_a\Psi_i(\eta\vec U_i^{(t)})\ut_i[a]^2
    \leq
    6\eta^2
    \norm{\ut_i-\vec u_i^{(t-1)}}_\infty^2.
\]
Because the learning rate is fixed, summing over
$t=1,\ldots,T$ telescopes the potential terms,
\[
    \sum_{t=1}^T
    \left(
       \Psi_i(\eta\vec U_i^{(t+1)})
       -
       \Psi_i(\eta\vec U_i^{(t)})
    \right)
    =
    \Psi_i(\eta\vec U_i^{(T+1)})-\Psi_i(\zero).
\]
Using $\Psi_i(\zero)<3c$ from \Cref{lem:calculus} proves
\eqref{eq:budget}.
\end{proof}

Both terms on the left of \eqref{eq:budget} are nonnegative. In the
next step, the gradient-weighted square provides the control on strategy
movement needed for the self-play analysis.

\subsection{From Weighted Squares to Strategy Movement}
\label{sec:analysis-movement}

We now show that the gradient-weighted square in \eqref{eq:budget}
controls the movement of the player's strategy. Since \algname{} may
assign very small probability to some actions, we work in square-root
coordinates and measure movement by $\norm{\sqrt{\vx}-\sqrt{\vx'}}_2^2$.
As motivated in \Cref{sec:normalization}, the stability built into the
potential is precisely what allows the gradient-weighted square to
control this movement.

\begin{lemma}\label{lem:movement}
Suppose player $i$ uses \algname{} with a fixed rate
$0<\eta\leq1/16$ and receives any sequence of utility vectors in
$[0,1]^{|\mathcal A_i|}$ after play. For every $t\geq1$,
\[
    \norm{\sqrt{\vx_i^{(t+1)}}-\sqrt{\xt_i}}_2^2
    \leq
    4\eta^2\sum_a
       \partial_a\Psi_i(\eta\vec U_i^{(t)})\ut_i[a]^2
    +16\eta^2
       \norm{\ut_i-\vec u_i^{(t-1)}}_\infty^2.
    \numberthis{eq:movement}
\]
\end{lemma}

\begin{proof}
Fix $t$. Between $\xt_i$ and $\vx_i^{(t+1)}$, both the cumulative
input and the optimistic correction change. We separate these two
effects using the intermediate distribution
\[
    \frac{
       a_c(\eta\vec U_i^{(t)}[a])
       (1+4\eta\ut_i[a])
    }{
       \sum_b
       a_c(\eta\vec U_i^{(t)}[b])
       (1+4\eta\ut_i[b])
    }.
\]
This distribution keeps the old cumulative input
$\eta\vec U_i^{(t)}$ but uses the new correction $\ut_i$.

By \Cref{lem:correction-change}, changing only the correction contributes
at most
\[
    8\eta^2
    \norm{\ut_i-\vec u_i^{(t-1)}}_\infty^2
\]
to the squared distance between the square-root vectors. By
\Cref{lem:input-change}, changing the cumulative input from
$\eta\vec U_i^{(t)}$ to $\eta\vec U_i^{(t+1)}$ while keeping the
correction fixed contributes at most
\[
    2\eta^2
    \sum_a
       \partial_a\Psi_i(\eta\vec U_i^{(t)})\ut_i[a]^2.
\]
Both estimates are proved in \Cref{app:movement} by differentiating
normalized paths and integrating their square-root speeds.

Applying the triangle inequality to the square-root vectors and then
$(x+y)^2\leq2x^2+2y^2$ gives
\[
    \norm{\sqrt{\vx_i^{(t+1)}}-\sqrt{\xt_i}}_2^2
    \leq
    4\eta^2\sum_a
       \partial_a\Psi_i(\eta\vec U_i^{(t)})\ut_i[a]^2
    +16\eta^2
       \norm{\ut_i-\vec u_i^{(t-1)}}_\infty^2.
\]
This proves \eqref{eq:movement}. In particular, the same
gradient-weighted square that appears in the potential bound
\eqref{eq:budget} controls strategy movement.
\end{proof}

\subsection{The Potential-level RVU Inequality}
\label{sec:analysis-rvu}

We now combine the cumulative potential bound \eqref{eq:budget} with the
movement estimate \eqref{eq:movement}. The gradient-weighted square in
\eqref{eq:budget} provides exactly the budget needed to control
accumulated squared strategy movement.

\begin{proposition}[Potential-level RVU inequality]\label{prop:rvu}
For a fixed rate $0<\eta\leq1/16$, any sequence of utility vectors in
$[0,1]^{|\mathcal A_i|}$ observed after play, and every integer
$T\geq1$,
\[
    \Psi_i(\eta\vec U_i^{(T+1)})
    \leq
    3c
       +10\eta^2\sum_{t=1}^T
          \norm{\ut_i-\vec u_i^{(t-1)}}_\infty^2 - \frac14\sum_{t=2}^T
          \norm{\sqrt{\xt_i}-\sqrt{\vx_i^{(t-1)}}}_2^2.
\]
\end{proposition}

\begin{proof}
Suppose first that $T\geq2$. Summing \eqref{eq:movement} over
$t=1,\ldots,T-1$, dividing by four, and extending the two nonnegative
sums on the right through round $T$ gives
\[
    \frac14\sum_{t=2}^T
       \norm{\sqrt{\xt_i}-\sqrt{\vx_i^{(t-1)}}}_2^2
    \leq
    \eta^2\sum_{t=1}^T\sum_a
       \partial_a\Psi_i(\eta\vec U_i^{(t)})\ut_i[a]^2
    +4\eta^2\sum_{t=1}^T
       \norm{\ut_i-\vec u_i^{(t-1)}}_\infty^2.
\]
Adding $\Psi_i(\eta\vec U_i^{(T+1)})$ and applying
\eqref{eq:budget} to the potential and weighted-square terms yields
\[
    \Psi_i(\eta\vec U_i^{(T+1)})
       +\frac14\sum_{t=2}^T
          \norm{\sqrt{\xt_i}-\sqrt{\vx_i^{(t-1)}}}_2^2
    \leq
    3c
       +10\eta^2\sum_{t=1}^T
          \norm{\ut_i-\vec u_i^{(t-1)}}_\infty^2.
\]
For $T=1$, the movement sum is empty, and the same conclusion follows
directly from \eqref{eq:budget}.
\end{proof}

This is the potential-level RVU inequality previewed in
\eqref{eq:rvu}. Its right-hand side depends only on the accumulated
squared prediction errors, while its left-hand side controls both the
terminal potential and accumulated squared strategy movement. The prediction error is expressed in terms of the centered utilities
$\ut_i$. Thus, it captures both changes in the utility vector $\nut_i$
and changes in the centering term $\ip{\xt_i}{\nut_i}$. 

No self-play assumption has been used so far. In the next subsection, we specialize to self-play in a fixed game and use the Hellinger prediction-error bound to close the argument.

\subsection{Prediction Error under Self-play in a Fixed Game}
\label{sec:analysis-game}

We now specialize to self-play in a fixed game. This is the first point
where the game structure enters the analysis. As discussed in
\Cref{sec:overview}, changes in the opponents' strategies change the
player's action-utility vector, while movement of the player's own
strategy changes the centering term. The following lemma controls both
effects by the players' squared Hellinger movement.

\begin{lemma}\label{lem:game}
In the model of \Cref{sec:game}, for every player $i$ and every $t\geq2$,
\[
    \norm{\ut_i-\vec u_i^{(t-1)}}_\infty^2
    \leq
    5\sum_{j=1}^n
       \norm{\sqrt{\xt_j}-\sqrt{\vx_j^{(t-1)}}}_2^2.
\]
\end{lemma}

\begin{proof}
Fix $i$ and $t\geq2$. We first control the change in action utilities
caused by the opponents. For any $a,b\in\mathcal A_i$, the function
\[
    \mathcal U_i(a,\vec s_{-i})-\mathcal U_i(b,\vec s_{-i})
\]
takes values in $[-1,1]$. Applying the expectation comparison from
\Cref{lem:expectation-comparison} to the opponents' current and previous
product distributions gives
\[
    \left|
       (\nut_i[a]-\nut_i[b])
       -
       (\vec\nu_i^{(t-1)}[a]-\vec\nu_i^{(t-1)}[b])
    \right|
    \leq
    2
    \left\|
       \sqrt{\bigotimes_{j\ne i}\xt_j}
       -
       \sqrt{\bigotimes_{j\ne i}\vx_j^{(t-1)}}
    \right\|_2.
\]
Since the same payoff function is evaluated on both rounds, the
fixed-game assumption applies. Using the product comparison from
\Cref{lem:product-comparison},
\[
    \left|
       (\nut_i[a]-\nut_i[b])
       -
       (\vec\nu_i^{(t-1)}[a]-\vec\nu_i^{(t-1)}[b])
    \right|
    \leq
    2
    \left(
       \sum_{j\ne i}
       \norm{\sqrt{\xt_j}-\sqrt{\vx_j^{(t-1)}}}_2^2
    \right)^{1/2}.
\]

We next account for centering. Adding and subtracting
$\ip{\xt_i}{\vec\nu_i^{(t-1)}}$ gives
\[
    \ut_i[a]-\vec u_i^{(t-1)}[a]
    =
    \nut_i[a]-\vec\nu_i^{(t-1)}[a]
    -
    \ip{\xt_i}{\nut_i-\vec\nu_i^{(t-1)}}
    -
    \ip{\xt_i-\vx_i^{(t-1)}}{\vec\nu_i^{(t-1)}}.
\]
The first three terms can be written as
\[
    \sum_b\xt_i[b]
    \left[
       (\nut_i[a]-\nut_i[b])
       -
       (\vec\nu_i^{(t-1)}[a]-\vec\nu_i^{(t-1)}[b])
    \right],
\]
so their absolute value is bounded by the preceding opponent-movement
estimate.

For the remaining centering term,
$\ip{\xt_i-\vx_i^{(t-1)}}{\one}=0$, so
\[
    \left|
       \ip{\xt_i-\vx_i^{(t-1)}}{\vec\nu_i^{(t-1)}}
    \right|
    \leq
    \frac12\norm{\xt_i-\vx_i^{(t-1)}}_1
    \leq
    \norm{\sqrt{\xt_i}-\sqrt{\vx_i^{(t-1)}}}_2.
\]
Combining the two contributions and maximizing over $a$ yields
\[
    \norm{\ut_i-\vec u_i^{(t-1)}}_\infty
    \leq
    2
    \left(
       \sum_{j\ne i}
       \norm{\sqrt{\xt_j}-\sqrt{\vx_j^{(t-1)}}}_2^2
    \right)^{1/2}
    +
    \norm{\sqrt{\xt_i}-\sqrt{\vx_i^{(t-1)}}}_2.
\]
Applying Cauchy--Schwarz to the two terms, with coefficients $2$ and
$1$, gives
\[
    \norm{\ut_i-\vec u_i^{(t-1)}}_\infty^2
    \leq
    5\sum_{j=1}^n
       \norm{\sqrt{\xt_j}-\sqrt{\vx_j^{(t-1)}}}_2^2.
\]
\end{proof}

On the first round, $\vec u_i^{(0)}=\zero$ and
$\norm{\vec u_i^{(1)}}_\infty\leq1$. Therefore,
\[
    \sum_{t=1}^T
       \norm{\ut_i-\vec u_i^{(t-1)}}_\infty^2
    \leq
    1
    +
    5\sum_{j=1}^n\sum_{t=2}^T
       \norm{\sqrt{\xt_j}-\sqrt{\vx_j^{(t-1)}}}_2^2.
\]
Substituting this bound into \Cref{prop:rvu} gives
\[
    \Psi_i(\eta\vec U_i^{(T+1)})
    +
    \frac14\sum_{t=2}^T
       \norm{\sqrt{\xt_i}-\sqrt{\vx_i^{(t-1)}}}_2^2
    \leq
    3c+10\eta^2
    +
    50\eta^2
    \sum_{j=1}^n\sum_{t=2}^T
       \norm{\sqrt{\xt_j}-\sqrt{\vx_j^{(t-1)}}}_2^2.
    \numberthis{eq:self-play}
\]

The product-distribution comparison is the key game-level step. It
controls the prediction error by a sum of squared strategy movements,
without an additional factor of $n$ inside the inequality. Thus, after
summing \eqref{eq:self-play} over players, only one factor of $n$
appears. This is what permits the learning rate to scale as
$\eta=O(1/\sqrt n)$.

\subsection{Uniform Movement and Regret under Self-play}
\label{sec:analysis-closure}

We now complete the self-play argument with
$\eta=1/(32\sqrt n)$. The final step has the same absorption structure
as in standard RVU analyses: after summing the playerwise inequalities,
the prediction-error term is controlled by the same accumulated movement
that appears on the left. Here, however, this movement is measured by
the squared Hellinger path length, 
\[
    \sum_{j=1}^n\sum_{t=2}^T
       \norm{\sqrt{\xt_j}-\sqrt{\vx_j^{(t-1)}}}_2^2.
\]
Our choice of learning rate leaves a positive coefficient on this term,
yielding a uniform bound on the accumulated squared movement. Feeding
this bound back into the playerwise inequality then controls each
terminal potential, and hence individual regret.

\begin{proposition}[Path Length]
\label{prop:closure}
Suppose every player uses \algname{} with
$\eta=1/(32\sqrt n)$ in the same fixed game. Then, for every
integer $T\geq1$,
\[
    \sum_{j=1}^n\sum_{t=2}^T
       \norm{\sqrt{\xt_j}-\sqrt{\vx_j^{(t-1)}}}_2^2
    \leq
    16nc,
\]
and, for every player $i$,
\[
    \Psi_i(\eta\vec U_i^{(T+1)})\leq4c.
\]
\end{proposition}

\begin{proof}
Fix a finite horizon $T$. Summing \eqref{eq:self-play} over
$i=1,\ldots,n$ and collecting the movement terms gives
\[
    \sum_{i=1}^n\Psi_i(\eta\vec U_i^{(T+1)})
    +
    \left(\frac14-50n\eta^2\right)
    \sum_{j=1}^n\sum_{t=2}^T
       \norm{\sqrt{\xt_j}-\sqrt{\vx_j^{(t-1)}}}_2^2
    \leq
    3nc+10n\eta^2.
\]
Since each potential is positive, we may omit the first term on the
left. Moreover,
\[
    n\eta^2=\frac1{1024},
    \qquad
    \frac14-50n\eta^2
      =
      \frac{103}{512}
      >
      \frac15.
\]
It follows that
\[
    \sum_{j=1}^n\sum_{t=2}^T
       \norm{\sqrt{\xt_j}-\sqrt{\vx_j^{(t-1)}}}_2^2
    \leq
    15nc+50n\eta^2
    \leq
    16nc.
\]

We now return to \eqref{eq:self-play} for a fixed player $i$. Using
the movement bound above and the nonnegativity of the player's movement
term,
\[
    \Psi_i(\eta\vec U_i^{(T+1)})
    \leq
    3c+10\eta^2+50\eta^2(16nc)
    =
    3c+10\eta^2+\frac{25}{32}c.
\]
Since $10\eta^2\leq10/1024<c/32$,
\[
    \Psi_i(\eta\vec U_i^{(T+1)})
    \leq
    \left(
       3+\frac1{32}+\frac{25}{32}
    \right)c
    =
    \frac{61}{16}c
    <
    4c.
\]
Both bounds therefore hold for every player and every finite horizon.
\end{proof}

The potential bound immediately yields the individual regret guarantee
from \Cref{thm:main}.

\begin{theorem}[Individual regret]
\label{thm:self-play}
Consider a fixed $n$-player finite game with at most $d\geq2$ actions
per player and utilities in $[0,1]$. Suppose every player uses
\algname{} (\Cref{alg:learning}) with $c=2+\log d$ and
$\eta=1/(32\sqrt n)$, and observes its exact expected-utility vector
$\nut_i$ after each simultaneous play. Then, for every player $i$ and
every integer $T\geq1$,
\[
    \Reg_i^{(T)}
    \leq
    96\sqrt n\,(2+\log d).
\]
The algorithm is deterministic, uses only the player's own past
observations, and does not require knowledge of the horizon.
\end{theorem}

\begin{proof}
By the potential certificate in \Cref{lem:calculus}, for every action
$a$,
\[
    c+\eta\vec U_i^{(T+1)}[a]
    \leq
    \Psi_i(\eta\vec U_i^{(T+1)})
    \leq
    4c.
\]
Hence,
\[
    \eta\Reg_i^{(T)}
      =
      \eta\max_a\vec U_i^{(T+1)}[a]
      \leq
      3c,
\]
where we used \eqref{eq:two-identities}. Substituting
$\eta=1/(32\sqrt n)$ and $c=2+\log d$ gives
\[
    \Reg_i^{(T)}
    \leq
    96\sqrt n\,(2+\log d).
\]

The response on round $t$ depends only on
$\vec U_i^{(t)}$, $\vec u_i^{(t-1)}$, and the fixed public parameters,
so it is deterministic, uncoupled, and independent of $T$. If
$|\mathcal A_i|=1$, the unique action is always played and
$\ut_i=\zero$ on every round.
\end{proof}

The factor $2+\log d$ comes from the potential scale $c$, while the $\sqrt n$ dependence comes from the choice $\eta=1/(32\sqrt n)$.

\subsection{The Equilibrium Guarantee}
\label{sec:analysis-equilibrium}

The individual regret bound immediately yields the corresponding
coarse-correlated-equilibrium guarantee for the average distribution of
play.

\begin{repeatcorollary}{cor:cce}[CCE]
Under the assumptions of \Cref{thm:self-play}, the distribution
\[
    \overline\sigma^{(T)}
      =
      \frac1T\sum_{t=1}^T
         \bigotimes_{j=1}^n\xt_j
\]
is a $96\sqrt n\,(2+\log d)/T$-coarse correlated equilibrium. Hence, no
player can improve its expected utility by more than
$96\sqrt n\,(2+\log d)/T$ by committing to any fixed action.
\end{repeatcorollary}

This completes the fixed-game self-play analysis. For arbitrary utility
sequences, the potential-level RVU inequality remains valid, but the
fixed-game prediction-error bound from \Cref{lem:game} no longer
applies. The separate adversarial guarantee obtained through the
learning-rate safeguard is proved in \Cref{app:safeguard}.

\section{Conclusion}

We introduced \algname{}, an uncoupled one-step optimistic learning rule
with $\mathcal{O}(\sqrt n\log d)$ individual regret in finite
general-sum self-play. Several questions remain open. The most basic is whether the
$\sqrt n$ dependence is optimal. More generally, sharp lower bounds for
uncoupled learning in general games are largely missing. It is also
natural to ask whether horizon-independent individual regret is possible
under substantially weaker feedback, such as bandit or noisy utility
observations. Finally, it would be interesting to understand whether
one-step multiplicative optimism can simultaneously yield constant regret and stronger
last-iterate convergence guarantees in structured classes of games.

\section*{Acknowledgments}

The authors thank Gabriele Farina for insightful discussions on regret
matching and multiplicative optimism, and especially for pointing out
the connection in \Cref{remark:OMWU} between our construction under entropic geometry and
Optimistic Multiplicative Weights Update.

\bibliographystyle{plainnat}
\bibliography{refs}

\clearpage

\appendix

\section{Related Work}\label{sec:related}

The study of learning in games has developed through several closely
related viewpoints. Fictitious play uses past play to predict an
opponent's future behavior, while no-regret learning evaluates a
player's cumulative performance against fixed actions in hindsight
\citep{robinson1951iterative,hannan1957approximation,fudenberg1998theory}.
Blackwell approachability gives a geometric way to study repeated
vector-payoff problems \citep{blackwell1956analog}.
\citet{abernethy2011blackwell} establish reductions between
approachability and no-regret learning, while
\citet{perchet2014approachability} develops further connections with
calibration. These links make cumulative advantage vectors and their
potentials natural objects in the analysis of learning in games.
Recent work strengthens this connection by giving reductions that
preserve quantitative convergence rates, so bounds proved in one
formulation can be transferred to another without losing their
dependence on the horizon and problem parameters
\citep{dann2025rate}.

The notion of regret determines which equilibrium deviations are
controlled. Correlated equilibrium, introduced by
\citet{aumann1974subjectivity}, allows a player to condition its
deviation on the action it is recommended. Calibration and stronger
notions of regret, such as internal and swap regret, provide learning
procedures for this equilibrium concept
\citep{foster1997calibrated,hart2000simple}. External regret controls
only deviations to a fixed action chosen independently of the
recommendation, and therefore leads to coarse correlated equilibrium
\citep{moulin1978strategically,blum2007external}. Our results concern external regret and CCE. 

Multiplicative weights provides a particularly clear bridge between
learning and computation. In zero-sum games, adaptive multiplicative
updates can be used to compute approximate equilibria
\citep{freund1999adaptive}, while closely related ideas also underlie
boosting and a broad range of optimization methods
\citep{freund1997decision,arora2012multiplicative}. The possibility of
learning faster in self-play than against an arbitrary sequence was
already studied in zero-sum games \citep{daskalakis2011near}. General games are more delicate because there is no common saddle-point
objective whose decrease simultaneously controls every player. Each
player must instead obtain its own regret guarantee from its own
observations. Our analysis reflects this separation. We maintain a
nonnegative potential for each player and combine the resulting
playerwise inequalities directly, without reducing the game to a single
global optimization problem.

Prediction-sensitive regret bounds provide another important ingredient.
\citet{chiang2012online} show that regret can improve when consecutive
loss functions vary gradually, while
\citet{rakhlin2013online,rakhlin2013optimization} develop optimistic
online learning around a prediction available before the next
observation. The prediction need not be the previous utility vector;
what matters in the regret bound is how accurately it predicts the
realized sequence. Related optimistic methods also exploit predictability in saddle-point
problems. For example, \citet{mertikopoulos2019optimistic} study
optimistic mirror descent with an extra-gradient step, which queries the
gradient again at an intermediate point before forming the next iterate.
Such oracle access is different from the simultaneous uncoupled setting
considered here, where a player chooses its strategy before observing
the current utility vector.

\citet{syrgkanis2015fast} connect recency-biased regularized learning
to faster rates in general games through RVU inequalities. Their bounds
couple utility variation with a negative strategy-movement term. This
coupling is especially useful in self-play because one player's
movement determines how much the utility vectors of the other players
can change. A related but distinct line of work studies what no-regret dynamics
imply for welfare. Under suitable smoothness conditions, welfare
guarantees extend from equilibrium to no-regret outcomes
\citep{roughgarden2015intrinsic}. \citet{foster2016learning} obtain
faster convergence results using approximate regret, where the
comparison with a fixed action is relaxed multiplicatively. 

The regret-matching line takes a different starting point from
regularized learning on the strategy simplex. Regret matching builds the
response directly from actionwise cumulative regret
\citep{hart2001general}, while potential-based analyses relate the
evolution of these cumulative quantities to regret guarantees
\citep{cesa2003potential}. This viewpoint also underlies
counterfactual regret minimization, which decomposes regret in
extensive-form games into local regret-minimization problems
\citep{zinkevich2007regret}. Several later variants improve practical
performance by giving recent observations more influence, including
discounted regret minimization and predictive regret matching
\citep{brown2019discounted,farina2021faster}. These two forms of recency are conceptually different. Discounting
changes the contribution of past observations to the accumulated
regret, whereas predictive methods retain the accumulated history and
use recent information to anticipate the next update. \algname{} follows
the latter principle. Its cumulative centered-utility vector keeps every
past observation with its original weight, while the preceding centered
utility enters separately through the multiplicative optimistic
correction.

Predictive Blackwell approachability provides a useful point of
comparison. \citet{farina2021faster} connect regret matching and regret
matching$^+$ with FTRL and mirror descent through an approachability
formulation, and derive predictive variants from this connection. In
those methods, the prediction enters additively by shifting the
cumulative regret before the response is computed. \algname{} uses the previous observation in a different way. The
cumulative state is left unchanged, and the prediction instead modifies
the action weights multiplicatively before normalization. This choice is
closely tied to our potential-based analysis. We construct the potential
so that the effect of the optimistic correction, the curvature of the
potential, and the resulting strategy movement can all be controlled
within the same argument. In this sense, multiplicative optimism and
the potential are designed together.

Normalization introduces another difficulty. A collection of positive
action weights can change only slightly in absolute terms while the
normalized strategy changes substantially, especially when the total
weight becomes small. This issue also appears in regret-matching
dynamics. \citet{farina2023regret} identify instability in regret
matching$^+$ and its predictive variant and study modifications that
restore stability, while \citet{zhang2025scale} develop scale-invariant
predictive regret matching with optimal average convergence in zero-sum
games and extensions to other structured settings. For our analysis, stability must hold even when some action weights
become very small. We therefore construct the potential so that the
relative sensitivity of these weights decreases together with their
total mass. This allows normalization to preserve the movement control
needed in the regret analysis without resetting or discounting the cumulative state.

Lifted regularization offers another route to individual regret.
LRL-OFTRL achieves logarithmic regret in general convex games by using
a nonnegative regret formulation and logarithmic regularization
\citep{farina2022near}. Cautious Optimism interprets acceleration
through dynamic pacing and introduces intrinsic Lipschitzness to
relate regularizer variation to the geometry of its divergence
\citep{soleymani2025cautious}. Its diverse-self-play and convex-game
guarantees have broader scope than our guarantees. We instead
construct one specific potential to improve the joint dependence on
$n$, $d$, and $T$. A player's ordinary regret may remain negative
throughout our proof. The positive quantity that survives summation in our proof 
is its potential.

The order of a prediction and the order of an analysis are distinct.
\citet{chen2020hedging} improve Optimistic Hedge in two-player games,
and \citet{daskalakis2021near} prove a multiplayer polylogarithmic
bound by analyzing higher-order discrete differences of the trajectory.
Their algorithm still uses a one-step predictor. ECHO-OFTRL and HOOD
instead incorporate higher-order filtered predictions into the updates
themselves \citep{liu2026constant,abbadi2026constant}. Our comparison
therefore concerns both the observed rates and the mechanisms used to
obtain them. The temporal error in our proof is only
$\ut_i-\vec u_i^{(t-1)}$. The additional control comes from the
potential and the Hellinger comparison.

The dependence on the number of players is similarly tied to the
specific comparison used in the proof. Standard game-level
$\ell_1$ estimates can introduce two factors of order $n$ after
squaring and summing, as in RVU and path-length analyses
\citep{syrgkanis2015fast,anagnostides2022last,anagnostides2022uncoupled,farina2022near,farina2021faster}.
Those estimates do not prove that linear player dependence is necessary
for every learning rule. Our potential controls squared Hellinger
movement, whose product inequality avoids one of the two factors that leads to $\sqrt{n}$ improvement.

Several nearby results obtain fast convergence under a different
protocol or for a different notion of regret. Clairvoyant MWU
\citep{piliouras2022beyond} starts from an implicit update in which the
next strategy is evaluated against the next strategy profile itself.
Its uncoupled implementation recovers a bounded-regret guarantee on a
sparse subsequence of the realized play, leading to fast convergence to
CCE. This is a different guarantee from ordinary external regret on the
full sequence of simultaneous rounds considered here.

A separate line of work strengthens the class of deviations against
which the learner competes. \citet{anagnostides2022uncoupled} obtain
$\mathcal{O}(\log T)$ swap regret in multiplayer games using optimistic
regularized learning and self-concordant barriers. Their analysis
controls a second-order path length of the dynamics and also retains an
$\mathcal{O}(\sqrt T)$ adversarial guarantee. More recently,
\citet{tsuchiya2026sublogarithmic} obtain sublogarithmic swap regret by
combining entropy and log-barrier regularization with a sensitivity
bound for the stationary distributions arising in the swap-regret
reduction. These guarantees imply convergence to correlated equilibrium,
which controls a richer class of deviations than the CCE guarantee
associated with external regret.

Recent work has also explored intermediate deviation classes between
external and swap regret. \citet{ahunbay2024first} study
first-order equilibrium notions in smooth games and characterize the
continuous strategy modifications controlled by projected gradient
ascent through tangent gradient fields. In normal-form games, these
guarantees include deviation classes richer than external regret and
lead to refinements such as semicoarse correlated equilibrium.
\citet{cai2025proximal} introduce proximal regret, which lies strictly
between external and swap regret, and show that ordinary gradient
descent already achieves sublinear proximal regret. The resulting
proximal correlated equilibria form a refinement of CCE.
\citet{soleymani2026exact} give a broader geometric characterization
for mirror descent, and FTRL. Their exact-form
framework identifies the deviation classes controlled by each method
through the geometry of the corresponding update, with proximal
deviations appearing as a special case.

Finally, fast self-play is useful as an individual learning prescription
only if a player also has protection when its opponents behave
differently. Adaptive wrappers already appear in fast optimistic
learning \citep{syrgkanis2015fast,daskalakis2021near}. HOOD also
provides an adversarial extension that permanently switches to another
learner after a threshold violation
\citep[Appendix~D]{abbadi2026constant}. Our distinction is the
rate-only safeguard. A violation
of the potential threshold decreases only the scalar learning rate.
No cumulative observation is reset, and the update remains within the
same algorithm albeit with a smaller learning rate. The threshold is never violated in the
self-play, so the protected rule preserves every iterate of the constant-regret trajectory in the self-play.

\section{Potential Geometry and the Taylor Bound}\label{app:calculus}

This appendix proves the analytic properties of the potential used in
\Cref{sec:analysis-geometry}. We proceed from the scalar function to the
geometry of the aggregated potential, then to multiplicative stability
and the finite-step Taylor bound.

Throughout the proofs, once a player $i$ is fixed, we suppress its index
on $\Psi_i$ and on the game vectors. All sums over actions are over
$\mathcal A_i$. The vector $\vec U$ denotes an arbitrary input to the
potential; along the learning trajectory, this input is
$\eta\vec U^{(t)}$.

\subsection{Scalar Properties, Derivatives, and Curvature}
\label{app:calculus-geometry}

\begin{repeatlemma}{lem:calculus}
Fix a player $i$ and let $c=2+\log d$. The potential $\Psi_i$ is
convex and continuously differentiable on $\R^{|\mathcal A_i|}$, with
locally Lipschitz gradient. For every $\vec U$,
\[
    \Psi_i(\vec U)>0,
    \qquad
    \Psi_i(\zero)
      =c|\mathcal A_i|^{1/(c-1)}
      <3c,
    \qquad
    \Psi_i(\vec U)
      \geq c+\max_a\vec U[a].
\]
Moreover,
\[
    \partial_a\Psi_i(\vec U)
      =
      \left(\frac{\Psi_i(\vec U)}c\right)^{2-c}
      a_c(\vec U[a])
      >0,
    \qquad
    \sum_a\partial_a\Psi_i(\vec U)\leq3.
\]
Where the Hessian exists,
\[
    0
      \preceq
      \nabla^2\Psi_i(\vec U)
      \preceq
      \diag\bigl(\partial_a\Psi_i(\vec U)\bigr).
\]
The same directional second-derivative bound holds almost everywhere
along every line segment.
\end{repeatlemma}

\begin{proof}
\emph{Scalar properties.}
Recall that
\[
    f(z)=
    \begin{cases}
       (1-z)^{-1},&z\leq0,\\
       1+z,&z\geq0,
    \end{cases}
    \qquad
    f'(z)=
    \begin{cases}
       (1-z)^{-2},&z\leq0,\\
       1,&z\geq0.
    \end{cases}
\]
The two branches agree in value and first derivative at zero. Hence,
$f$ is positive and continuously differentiable, with
$0<f'(z)\leq1$. Away from zero,
\[
    f''(z)=
    \begin{cases}
       2(1-z)^{-3},&z<0,\\
       0,&z>0.
    \end{cases}
\]
Thus, $f$ is convex and $f'$ is globally $2$-Lipschitz. We will use
the identity
\[
    f'(z)=\min\{f(z)^2,1\}.
\]
We will also use the lower bound $f(z)\geq1+z$. It is immediate for
$z\geq0$, while for $z\leq0$,
\[
    f(z)-(1+z)
      =
      \frac{z^2}{1-z}
      \geq0.
\]

\emph{Convexity and the potential certificate.}
Since $c-1>1$, the $\ell_{c-1}$ norm is convex and coordinatewise
nondecreasing on the nonnegative orthant. For
$0\leq\alpha\leq1$, convexity of $f$ gives, coordinatewise,
\[
    f\!\left(
       \frac{\alpha\vec U[a]+(1-\alpha)\vec V[a]}c
    \right)
    \leq
    \alpha f\!\left(\frac{\vec U[a]}c\right)
    +(1-\alpha)f\!\left(\frac{\vec V[a]}c\right).
\]
Monotonicity and convexity of the norm therefore imply
\[
    \Psi(\alpha\vec U+(1-\alpha)\vec V)
    \leq
    \alpha\Psi(\vec U)+(1-\alpha)\Psi(\vec V),
\]
so $\Psi$ is convex.

Since the norm dominates each of its coordinates, for every action
$a$,
\[
    \Psi(\vec U)
      \geq
      c f\!\left(\frac{\vec U[a]}c\right)
      \geq
      c+\vec U[a].
\]
Taking the maximum over $a$ gives the stated lower bound. Moreover,
$f(0)=1$, so
\[
    \Psi(\zero)
      =
      c|\mathcal A_i|^{1/(c-1)}.
\]
Since $|\mathcal A_i|\leq d$ and $c-1=1+\log d$,
\[
    |\mathcal A_i|^{1/(c-1)}
    \leq
    d^{1/(1+\log d)}
    =
    \exp\!\left(\frac{\log d}{1+\log d}\right)
    <e<3.
\]
Hence $\Psi(\zero)<3c$. Positivity follows immediately from $f>0$.

\emph{Gradient and total derivative mass.}
Direct differentiation of the potential gives
\[
    \partial_a\Psi(\vec U)
      =
      \left(\frac{\Psi(\vec U)}c\right)^{2-c}
      f\!\left(\frac{\vec U[a]}c\right)^{c-2}
      f'\!\left(\frac{\vec U[a]}c\right).
\]
Using $f'(z)=\min\{f(z)^2,1\}$, the action-dependent factor is
\[
    \begin{cases}
       \left(1-\frac{\vec U[a]}c\right)^{-c},
          &\vec U[a]\leq0,\\
       \left(1+\frac{\vec U[a]}c\right)^{c-2},
          &\vec U[a]\geq0.
    \end{cases}
\]
This is exactly $a_c(\vec U[a])$, and therefore
\[
    \partial_a\Psi(\vec U)
      =
      \left(\frac{\Psi(\vec U)}c\right)^{2-c}
      a_c(\vec U[a])
      >0.
\]

To bound the total derivative mass, use $f'\leq1$ and H\"older's
inequality,
\[
    \sum_a
       f\!\left(\frac{\vec U[a]}c\right)^{c-2}
    \leq
    \left(
       \sum_a
       f\!\left(\frac{\vec U[a]}c\right)^{c-1}
    \right)^{(c-2)/(c-1)}
    |\mathcal A_i|^{1/(c-1)}.
\]
Since the power of the sum equals
$(\Psi(\vec U)/c)^{c-2}$, we obtain
\[
    \sum_a\partial_a\Psi(\vec U)
    \leq
    |\mathcal A_i|^{1/(c-1)}
    <3.
\]

\emph{Weighted curvature.}
For $z\ne0$,
\[
    \frac{a_c'(z)}{a_c(z)}
      =
      \begin{cases}
         (1-z/c)^{-1},&z<0,\\
         (c-2)/(c+z),&z>0.
      \end{cases}
\]
Both branches belong to $[0,1]$. At an input with no zero coordinates,
differentiating the gradient formula gives
\[
    \partial_b\partial_a\Psi(\vec U)
      =
      \mathbf{1}_{\{a=b\}}
      \partial_a\Psi(\vec U)
      \frac{a_c'(\vec U[a])}{a_c(\vec U[a])}
      -
      \frac{c-2}{\Psi(\vec U)}
      \partial_a\Psi(\vec U)\partial_b\Psi(\vec U).
\]
Equivalently,
\[
    \nabla^2\Psi(\vec U)
      =
      \diag\!\left(
         \partial_a\Psi(\vec U)
         \frac{a_c'(\vec U[a])}{a_c(\vec U[a])}
      \right)
      -
      \frac{c-2}{\Psi(\vec U)}
      \nabla\Psi(\vec U)\nabla\Psi(\vec U)^\top.
\]
The second term is positive semidefinite before subtraction. Hence, for
every $\vec v$,
\[
    \vec v^\top\nabla^2\Psi(\vec U)\vec v
    \leq
    \sum_a
       \partial_a\Psi(\vec U)
       \frac{a_c'(\vec U[a])}{a_c(\vec U[a])}
       \vec v[a]^2
    \leq
    \sum_a
       \partial_a\Psi(\vec U)\vec v[a]^2.
\]
Since $\Psi$ is convex, its Hessian is positive semidefinite wherever
it exists. This proves
\[
    0
    \preceq
    \nabla^2\Psi(\vec U)
    \preceq
    \diag\bigl(\partial_a\Psi(\vec U)\bigr)
\]
at every point where the Hessian exists.

\emph{Regularity at zero coordinates.}
It remains to justify the regularity needed when a coordinate crosses
the junction at zero. The scalar function $f$ is continuously
differentiable with Lipschitz derivative. On every compact set of
inputs, the values $f(\vec U[a]/c)$ are uniformly bounded away from
zero. Since the $\ell_{c-1}$ norm is smooth on the strictly positive
orthant, the gradient formula above is locally Lipschitz on all of
$\R^{|\mathcal A_i|}$.

Now fix a line $\vec U+s\vec v$. Every nonconstant coordinate crosses
zero at most once. Away from these finitely many crossing points, the
Hessian exists and the preceding quadratic-form bound applies. A
coordinate that is identically zero has $\vec v[a]=0$ and does not
contribute to the directional second derivative. Thus, almost
everywhere along the line,
\[
    0
    \leq
    \frac{d^2}{ds^2}\Psi(\vec U+s\vec v)
    \leq
    \sum_a
       \partial_a\Psi(\vec U+s\vec v)\vec v[a]^2.
\]
Finally, local Lipschitz continuity of $\nabla\Psi$ makes the first
derivative along the line absolutely continuous on compact intervals,
so this bound can be integrated across the zero crossings. No Hessian
value at a crossing is required.
\end{proof}

\subsection{Multiplicative Stability}
\label{app:calculus-multiplicative}

We next prove the multiplicative stability bounds used in the Taylor and
movement analyses. The first controls the gradient coordinates under a
finite change of the potential input. The second transfers this control
to the normalized probabilities of \algname{}.

\begin{repeatlemma}{lem:multiplicative}[Multiplicative stability]
For every $\vec U,\vec u\in\R^{|\mathcal A_i|}$, every $\eta>0$,
and every action $a\in\mathcal A_i$,
\[
    \left|
       \log
       \frac{\partial_a\Psi_i(\vec U+\eta\vec u)}
            {\partial_a\Psi_i(\vec U)}
    \right|
    \leq
    2\eta\norm{\vec u}_\infty.
\]
For \algname{} with a fixed rate $0<\eta\leq1/16$ and any sequence of
utility vectors in $[0,1]^{|\mathcal A_i|}$ observed after play,
\[
    \left|
       \log
       \frac{\vx_i^{(t+1)}[a]}{\xt_i[a]}
    \right|
    \leq
    2\eta\norm{\ut_i}_\infty
    +
    \frac{8\eta}{1-4\eta}
       \norm{\ut_i-\vec u_i^{(t-1)}}_\infty.
\]
In particular, if $\eta\leq1/32$, then
$\vx_i^{(t+1)}[a]/\xt_i[a]\in[1/2,2]$ for every action $a$.
\end{repeatlemma}

\begin{proof}
We first prove \eqref{eq:gradient-multiplicative}. By
\eqref{eq:weight-log-derivative}, the logarithmic derivative of $a_c$
has absolute value at most one on both open branches. Since the
branches of $\log a_c$ agree at zero, $\log a_c$ is globally
$1$-Lipschitz. Hence,
\[
    \left|
       \log
       \frac{a_c(\vec U[a]+\eta\vec u[a])}
            {a_c(\vec U[a])}
    \right|
    \leq
    \eta|\vec u[a]|.
\]

From the scalar formulas in the preceding subsection,
$0<f'(z)/f(z)\leq1$. Thus, for every action $a$,
\[
    e^{-\eta\norm{\vec u}_\infty/c}
    f\!\left(\frac{\vec U[a]}c\right)
    \leq
    f\!\left(\frac{\vec U[a]+\eta\vec u[a]}c\right)
    \leq
    e^{\eta\norm{\vec u}_\infty/c}
    f\!\left(\frac{\vec U[a]}c\right).
\]
Applying the coordinatewise monotonicity and homogeneity of the
$\ell_{c-1}$ norm gives
\[
    \left|
       \log
       \frac{\Psi(\vec U+\eta\vec u)}
            {\Psi(\vec U)}
    \right|
    \leq
    \frac{\eta}{c}\norm{\vec u}_\infty.
\]
Using the gradient factorization,
\[
    \log
    \frac{\partial_a\Psi(\vec U+\eta\vec u)}
         {\partial_a\Psi(\vec U)}
    =
    (2-c)
    \log
    \frac{\Psi(\vec U+\eta\vec u)}
         {\Psi(\vec U)}
    +
    \log
    \frac{a_c(\vec U[a]+\eta\vec u[a])}
         {a_c(\vec U[a])}.
\]
Therefore,
\[
    \left|
       \log
       \frac{\partial_a\Psi(\vec U+\eta\vec u)}
            {\partial_a\Psi(\vec U)}
    \right|
    \leq
    \left(
       \frac{c-2}{c}+1
    \right)
    \eta\norm{\vec u}_\infty
    \leq
    2\eta\norm{\vec u}_\infty,
\]
which proves \eqref{eq:gradient-multiplicative}.

We next prove the probability-ratio bound. Let
\[
    q_t[a]
      \coloneqq
      a_c(\eta\vec U^{(t)}[a])
      \bigl(1+4\eta\vec u^{(t-1)}[a]\bigr)
\]
denote the unnormalized weight on round $t$. Since
$\norm{\vec u^{(t-1)}}_\infty\leq1$ and $\eta\leq1/16$, all correction
factors are positive.

The cumulative update
$\vec U^{(t+1)}=\vec U^{(t)}+\ut$ and the Lipschitz bound for
$\log a_c$ give
\[
    \left|
       \log
       \frac{a_c(\eta\vec U^{(t+1)}[a])}
            {a_c(\eta\vec U^{(t)}[a])}
    \right|
    \leq
    \eta\norm{\ut}_\infty.
\]
Moreover, the derivative of $s\mapsto\log(1+4\eta s)$ on $[-1,1]$ is
at most $4\eta/(1-4\eta)$, so
\[
    \left|
       \log
       \frac{1+4\eta\ut[a]}
            {1+4\eta\vec u^{(t-1)}[a]}
    \right|
    \leq
    \frac{4\eta}{1-4\eta}
    \norm{\ut-\vec u^{(t-1)}}_\infty.
\]
Hence, if
\[
    B_t
      =
      \eta\norm{\ut}_\infty
      +
      \frac{4\eta}{1-4\eta}
      \norm{\ut-\vec u^{(t-1)}}_\infty,
\]
then every unnormalized weight satisfies
\[
    e^{-B_t}
    \leq
    \frac{q_{t+1}[a]}{q_t[a]}
    \leq
    e^{B_t}.
\]

Let $Q_t=\sum_a q_t[a]$. Since
\[
    \frac{Q_{t+1}}{Q_t}
      =
      \sum_a
         \xt[a]\frac{q_{t+1}[a]}{q_t[a]},
\]
the normalization ratio also belongs to $[e^{-B_t},e^{B_t}]$.
Consequently,
\[
    \left|
       \log
       \frac{\vx^{(t+1)}[a]}{\xt[a]}
    \right|
    \leq
    2B_t,
\]
which is exactly \eqref{eq:iterate-multiplicative}.

Finally, if $\eta\leq1/32$, then
$\norm{\ut}_\infty\leq1$ and
$\norm{\ut-\vec u^{(t-1)}}_\infty\leq2$. Thus,
\[
    2B_t
    \leq
    2\eta+\frac{16\eta}{1-4\eta}
    \leq
    \frac1{16}+\frac47
    =
    \frac{71}{112}
    <
    \log2.
\]
Exponentiating gives
$\vx^{(t+1)}[a]/\xt[a]\in[1/2,2]$ for every action $a$.
\end{proof}

\subsection{The Finite-step Taylor Estimate}
\label{app:calculus-taylor}

We now combine the weighted curvature bound with multiplicative
stability to obtain the finite-step Taylor estimate used in the main
analysis.

\begin{repeatlemma}{lem:taylor}
For every $\vec U,\vec u\in\R^{|\mathcal A_i|}$ with
$\norm{\vec u}_\infty\leq1$ and every $0<\eta\leq1/8$,
\[
    \Psi_i(\vec U+\eta\vec u)-\Psi_i(\vec U)
    \leq
    \eta\ip{\nabla\Psi_i(\vec U)}{\vec u}
    +
    \frac{2\eta^2}{3}
    \sum_a
       \partial_a\Psi_i(\vec U)\vec u[a]^2.
\]
\end{repeatlemma}

\begin{proof}
Consider the segment from $\vec U$ to $\vec U+\eta\vec u$,
parameterized by $\vec U+\theta\eta\vec u$ for
$0\leq\theta\leq1$. By the chain rule,
\[
    \frac{d}{d\theta}
       \Psi(\vec U+\theta\eta\vec u)
    =
    \eta
    \ip{\nabla\Psi(\vec U+\theta\eta\vec u)}{\vec u}.
\]
Where the Hessian exists, differentiating once more gives
\[
    \frac{d^2}{d\theta^2}
       \Psi(\vec U+\theta\eta\vec u)
    =
    \eta^2
    \vec u^\top
       \nabla^2\Psi(\vec U+\theta\eta\vec u)
    \vec u.
\]
The weighted curvature bound from \Cref{lem:calculus} therefore implies,
almost everywhere along the segment,
\[
    \frac{d^2}{d\theta^2}
       \Psi(\vec U+\theta\eta\vec u)
    \leq
    \eta^2
    \sum_a
       \partial_a\Psi(\vec U+\theta\eta\vec u)
       \vec u[a]^2.
\]
At this point the weights are evaluated at the intermediate input
$\vec U+\theta\eta\vec u$, whereas the desired Taylor bound uses the
weights at the initial input $\vec U$. This is where multiplicative
stability enters. By \eqref{eq:gradient-multiplicative},
\[
    \partial_a\Psi(\vec U+\theta\eta\vec u)
    \leq
    e^{2\theta\eta\norm{\vec u}_\infty}
    \partial_a\Psi(\vec U)
    \leq
    \frac43\partial_a\Psi(\vec U),
    \qquad
    0\leq\theta\leq1,
    \numberthis{eq:segment}
\]
where the last inequality follows from
$\norm{\vec u}_\infty\leq1$, $\eta\leq1/8$, and
$e^{1/4}<4/3$. Hence, throughout the segment,
\[
    \frac{d^2}{d\theta^2}
       \Psi(\vec U+\theta\eta\vec u)
    \leq
    \frac{4\eta^2}{3}
    \sum_a
       \partial_a\Psi(\vec U)\vec u[a]^2
\]
almost everywhere.

We now integrate this second-order bound along the segment. The
integral form of Taylor's theorem gives
\[
    \Psi(\vec U+\eta\vec u)-\Psi(\vec U)
      -\eta\ip{\nabla\Psi(\vec U)}{\vec u}
    =
    \int_0^1
       (1-\theta)
       \frac{d^2}{d\theta^2}
          \Psi(\vec U+\theta\eta\vec u)
    \,d\theta.
\]
Using the preceding bound and
$\int_0^1(1-\theta)\,d\theta=1/2$, we obtain
\[
    \Psi(\vec U+\eta\vec u)-\Psi(\vec U)
      -\eta\ip{\nabla\Psi(\vec U)}{\vec u}
    \leq
    \frac{2\eta^2}{3}
    \sum_a
       \partial_a\Psi(\vec U)\vec u[a]^2.
\]
Rearranging proves the lemma.

The argument remains valid when the segment crosses zero coordinates.
Indeed, \Cref{lem:calculus} shows that the first derivative along the
segment is absolutely continuous and that the weighted
second-derivative bound holds almost everywhere, which is sufficient
for the integral Taylor formula above. Taking
$\vec U=\eta\vec U^{(t)}$ and $\vec u=\ut$ gives the Taylor estimate
used in the round-by-round analysis.
\end{proof}

\section{Stability of the Normalized Response}\label{app:movement}

This appendix proves the response-stability estimates used in
\Cref{sec:analysis-movement}. We first establish
\eqref{eq:compensation}, which relates the total derivative mass
$\sum_a\partial_a\Psi(\vec U)$ to the logarithmic sensitivity of the
power weights $a_c(\vec U[a])$. We then derive a differentiation
identity for normalized responses and apply it to the two changes
between consecutive rounds. One path changes the optimistic correction
while keeping the cumulative input fixed, and the other changes the
cumulative input while keeping the correction fixed. The two paths meet
at the intermediate distribution used in the proof of
\Cref{lem:movement}.

Fix a player $i$ and suppress its index throughout the proofs. All sums
over actions are over $\mathcal A_i$, and square roots of probability
vectors are taken coordinatewise. In the path arguments below,
$\theta\in[0,1]$ parametrizes an interpolation between two responses.

\subsection{Compensation for Normalization}
\label{app:movement-compensation}

In \Cref{sec:potential-scalar}, the negative branch of $f$ was chosen
so that the squared relative sensitivity of a small weight decreases
together with the weight itself. After aggregation, normalization
depends on the total derivative mass
$\sum_a\partial_a\Psi(\vec U)$. We now prove the corresponding relation
for the gradient weights of $\Psi$ and the power weights $a_c$.

\begin{lemma}\label{lem:compensation}
For every input $\vec U\in\R^{|\mathcal A_i|}$,
\[
    \sum_a\partial_a\Psi(\vec U)
      \geq
      \frac13
      \min\!\left\{
         \left(\frac{\Psi(\vec U)}c\right)^2,1
      \right\}.
    \numberthis{eq:mass-lower}
\]
For each action $a$ with $\vec U[a]\ne0$,
\[
    \left(
       \frac{a_c'(\vec U[a])}{a_c(\vec U[a])}
    \right)^2
      \leq
      \min\!\left\{
         \left(\frac{\Psi(\vec U)}c\right)^2,1
      \right\}.
    \numberthis{eq:sensitivity}
\]
Consequently,
\[
    \frac{
       \bigl(a_c'(\vec U[a])/a_c(\vec U[a])\bigr)^2
    }{
       \sum_b\partial_b\Psi(\vec U)
    }
    \leq3
    \quad\text{where the derivatives exist}.
\]
\end{lemma}

\begin{proof}
We first lower bound the total derivative mass. Since the
$\ell_{c-1}$ norm defining $\Psi$ dominates each coordinate,
\[
    0
    <
    f\!\left(\frac{\vec U[a]}c\right)
    \leq
    \frac{\Psi(\vec U)}c
\]
for every action $a$. Using the derivative formula from
\Cref{lem:calculus} together with
$f'(z)=\min\{f(z)^2,1\}$ gives
\[
    \sum_a\partial_a\Psi(\vec U)
      =
      \left(\frac{\Psi(\vec U)}c\right)^{2-c}
      \sum_a
         f\!\left(\frac{\vec U[a]}c\right)^{c-2}
         \min\!\left\{
            f\!\left(\frac{\vec U[a]}c\right)^2,1
         \right\}.
\]

The coordinatewise bound above implies
\[
    \min\!\left\{
       f\!\left(\frac{\vec U[a]}c\right)^2,1
    \right\}
    \geq
    f\!\left(\frac{\vec U[a]}c\right)^2
    \min\!\left\{
       1,
       \left(\frac{\Psi(\vec U)}c\right)^{-2}
    \right\}.
\]
To see this, if $\Psi(\vec U)/c\leq1$, then
$f(\vec U[a]/c)\leq1$, and the two sides are equal. If
$\Psi(\vec U)/c>1$, the right-hand side equals
\[
    \frac{
       f(\vec U[a]/c)^2
    }{
       (\Psi(\vec U)/c)^2
    },
\]
which is at most both $f(\vec U[a]/c)^2$ and $1$.

Substituting this bound into the derivative sum yields
\[
    \sum_a\partial_a\Psi(\vec U)
    \geq
    \left(\frac{\Psi(\vec U)}c\right)^{2-c}
    \min\!\left\{
       1,
       \left(\frac{\Psi(\vec U)}c\right)^{-2}
    \right\}
    \sum_a
       f\!\left(\frac{\vec U[a]}c\right)^c.
\]
It remains to lower bound the final sum. H\"older's inequality gives
\[
    \sum_a
       f\!\left(\frac{\vec U[a]}c\right)^{c-1}
    \leq
    \left(
       \sum_a
          f\!\left(\frac{\vec U[a]}c\right)^c
    \right)^{(c-1)/c}
    |\mathcal A_i|^{1/c}.
\]
The left-hand side equals
$(\Psi(\vec U)/c)^{c-1}$. Raising the inequality to the power
$c/(c-1)$ and rearranging therefore gives
\[
    \sum_a
       f\!\left(\frac{\vec U[a]}c\right)^c
    \geq
    |\mathcal A_i|^{-1/(c-1)}
    \left(\frac{\Psi(\vec U)}c\right)^c.
\]
Combining the last two inequalities and simplifying the powers gives
\[
    \sum_a\partial_a\Psi(\vec U)
    \geq
    |\mathcal A_i|^{-1/(c-1)}
    \min\!\left\{
       \left(\frac{\Psi(\vec U)}c\right)^2,1
    \right\}.
\]
Since $|\mathcal A_i|^{1/(c-1)}<3$, this proves
\eqref{eq:mass-lower}.

We next bound the logarithmic sensitivity of $a_c$. If
$\vec U[a]<0$, then
\[
    \frac{a_c'(\vec U[a])}{a_c(\vec U[a])}
      =
      \left(1-\frac{\vec U[a]}c\right)^{-1}
      =
      f\!\left(\frac{\vec U[a]}c\right).
\]
On this branch,
$f(\vec U[a]/c)\leq1$, while the coordinate bound above gives
$f(\vec U[a]/c)\leq\Psi(\vec U)/c$. Hence,
\[
    \left(
       \frac{a_c'(\vec U[a])}{a_c(\vec U[a])}
    \right)^2
    \leq
    \min\!\left\{
       \left(\frac{\Psi(\vec U)}c\right)^2,1
    \right\}.
\]

If $\vec U[a]>0$, then
$f(\vec U[a]/c)>1$, which implies $\Psi(\vec U)/c>1$. Moreover,
\[
    0
    \leq
    \frac{a_c'(\vec U[a])}{a_c(\vec U[a])}
    =
    \frac{c-2}{c+\vec U[a]}
    \leq1.
\]
Since the minimum in \eqref{eq:sensitivity} is now equal to one, the
same bound follows on the positive branch.

Finally, dividing \eqref{eq:sensitivity} by
\eqref{eq:mass-lower} cancels the common factor
\[
    \min\!\left\{
       \left(\frac{\Psi(\vec U)}c\right)^2,1
    \right\}
\]
and gives precisely the compensation bound \eqref{eq:compensation}.
\end{proof}

At a coordinate with $\vec U[a]=0$, the two one-sided logarithmic
derivatives of $a_c$ may differ, but both satisfy the sensitivity bound
above. Along any affine path used below, a nonconstant coordinate
crosses zero at most once, while a coordinate that remains zero has
zero rate of change. Consequently, \eqref{eq:compensation} holds almost
everywhere along each path, which is sufficient for the integrations
below.

\subsection{Normalized Paths and Square-root Movement}
\label{app:movement-paths}

We next relate the change of a normalized response to the derivatives
of its action weights along a path. The key identity expresses the
squared speed of the square-root probabilities as a variance of
logarithmic derivatives. Integrating this speed then bounds the squared
movement between the endpoints of the path.

\begin{lemma}
\label{lem:normalized-path}
Let $\vx(\theta)$, $0\leq\theta\leq1$, be a Lipschitz path of
strictly positive probability vectors on $\mathcal A_i$. Then, for
almost every $\theta$,
\[
    \left\|
       \frac{d}{d\theta}\sqrt{\vx(\theta)}
    \right\|_2^2
    =
    \frac14
    \sum_a
       \vx(\theta)[a]
       \left(
          \frac{d}{d\theta}\log\vx(\theta)[a]
       \right)^2
    =
    \frac14
    \Var_{a\sim\vx(\theta)}
    \left[
       \frac{d}{d\theta}\log\vx(\theta)[a]
    \right].
    \numberthis{eq:normalized-speed}
\]
Moreover,
\[
    \norm{\sqrt{\vx(1)}-\sqrt{\vx(0)}}_2^2
    \leq
    \int_0^1
       \left\|
          \frac{d}{d\theta}\sqrt{\vx(\theta)}
       \right\|_2^2
    \,d\theta.
    \numberthis{eq:path-integral}
\]
\end{lemma}

\begin{proof}
Since every coordinate of $\vx(\theta)$ is continuous and strictly
positive on $[0,1]$, it is bounded away from zero along this fixed
path. Thus, the coordinatewise logarithms and square roots are
absolutely continuous, and the following derivatives exist almost
everywhere.

For each action $a$, the chain rule gives
\[
    \frac{d}{d\theta}\sqrt{\vx(\theta)[a]}
    =
    \frac12
    \sqrt{\vx(\theta)[a]}
    \frac{d}{d\theta}\log\vx(\theta)[a].
\]
Squaring and summing over actions gives
\[
    \left\|
       \frac{d}{d\theta}\sqrt{\vx(\theta)}
    \right\|_2^2
    =
    \frac14
    \sum_a
       \vx(\theta)[a]
       \left(
          \frac{d}{d\theta}\log\vx(\theta)[a]
       \right)^2.
\]

The logarithmic derivatives have mean zero under $\vx(\theta)$. Indeed,
\[
    \sum_a
       \vx(\theta)[a]
       \frac{d}{d\theta}\log\vx(\theta)[a]
    =
    \sum_a
       \frac{d}{d\theta}\vx(\theta)[a]
    =
    \frac{d}{d\theta}\sum_a\vx(\theta)[a]
    =
    0.
\]
Their second moment is therefore equal to their variance, which proves
\eqref{eq:normalized-speed}.

To compare the endpoints, absolute continuity gives
\[
    \sqrt{\vx(1)}-\sqrt{\vx(0)}
    =
    \int_0^1
       \frac{d}{d\theta}\sqrt{\vx(\theta)}
    \,d\theta.
\]
The triangle inequality followed by Cauchy--Schwarz yields
\[
    \norm{\sqrt{\vx(1)}-\sqrt{\vx(0)}}_2
    \leq
    \int_0^1
       \left\|
          \frac{d}{d\theta}\sqrt{\vx(\theta)}
       \right\|_2
    \,d\theta
    \leq
    \left(
       \int_0^1
          \left\|
             \frac{d}{d\theta}\sqrt{\vx(\theta)}
          \right\|_2^2
       \,d\theta
    \right)^{1/2}.
\]
Squaring proves \eqref{eq:path-integral}.
\end{proof}

We will apply the lemma to normalized positive action weights. Suppose
\[
    \vx(\theta)[a]
    =
    \frac{q_a(\theta)}
         {\sum_b q_b(\theta)},
    \qquad
    q_a(\theta)>0.
\]
Differentiating the logarithm gives
\[
    \frac{d}{d\theta}\log\vx(\theta)[a]
    =
    \frac{d}{d\theta}\log q_a(\theta)
    -
    \sum_b
       \vx(\theta)[b]
       \frac{d}{d\theta}\log q_b(\theta).
\]
Thus, normalization subtracts the $\vx(\theta)$-average logarithmic
derivative from every action. Since subtracting a common scalar does
not change variance, \eqref{eq:normalized-speed} becomes
\[
    \left\|
       \frac{d}{d\theta}\sqrt{\vx(\theta)}
    \right\|_2^2
    =
    \frac14
    \Var_{a\sim\vx(\theta)}
    \left[
       \frac{d}{d\theta}\log q_a(\theta)
    \right].
\]
The two path estimates below follow by substituting the corresponding
unnormalized weights $q_a(\theta)$ into this identity.

\subsection{Changing the Optimistic Correction}
\label{app:movement-correction}

We first compare $\xt$ with the intermediate distribution from the
proof of \Cref{lem:movement}. The cumulative input
$\eta\vec U^{(t)}$ remains fixed, while the optimistic correction
changes from $\vec u^{(t-1)}$ to $\ut$.

\begin{lemma}\label{lem:correction-change}
Fix a round $t$ and $0<\eta\leq1/16$. For $0\leq\theta\leq1$,
let
\[
    \vx(\theta)[a]
      =
      \frac{
         a_c(\eta\vec U^{(t)}[a])
         \bigl(
            1+4\eta
            ((1-\theta)\vec u^{(t-1)}[a]+\theta\ut[a])
         \bigr)
      }{
         \sum_b
         a_c(\eta\vec U^{(t)}[b])
         \bigl(
            1+4\eta
            ((1-\theta)\vec u^{(t-1)}[b]+\theta\ut[b])
         \bigr)
      }.
\]
Then
\[
    \norm{\sqrt{\vx(1)}-\sqrt{\vx(0)}}_2^2
    \leq
    8\eta^2
    \norm{\ut-\vec u^{(t-1)}}_\infty^2.
    \numberthis{eq:correction-movement}
\]
\end{lemma}

\begin{proof}
The path starts at $\vx(0)=\xt$. At $\theta=1$, the cumulative input
is still $\eta\vec U^{(t)}$, while the correction has changed from
$\vec u^{(t-1)}$ to $\ut$. Thus, $\vx(1)$ is exactly the intermediate
distribution used in the proof of \Cref{lem:movement}.

For every $\theta\in[0,1]$, the interpolated correction satisfies
\[
    \left|
       (1-\theta)\vec u^{(t-1)}[a]+\theta\ut[a]
    \right|
    \leq1.
\]
Since $\eta\leq1/16$,
\[
    \frac34
    \leq
    1+4\eta
       \bigl(
          (1-\theta)\vec u^{(t-1)}[a]+\theta\ut[a]
       \bigr)
    \leq
    \frac54.
\]
Hence all unnormalized weights remain strictly positive, and the path
is Lipschitz.

We now apply the normalized-path identity from
\Cref{lem:normalized-path}. The power weight
$a_c(\eta\vec U^{(t)}[a])$ is constant along this path, so the
logarithmic derivative of the unnormalized weight is
\[
    \frac{
       4\eta
       \bigl(\ut[a]-\vec u^{(t-1)}[a]\bigr)
    }{
       1+4\eta
       \bigl(
          (1-\theta)\vec u^{(t-1)}[a]+\theta\ut[a]
       \bigr)
    }.
\]
Using \eqref{eq:normalized-speed} and bounding the variance by its
second moment gives
\[
    \left\|
       \frac{d}{d\theta}\sqrt{\vx(\theta)}
    \right\|_2^2
    \leq
    \frac14
    \sum_a
       \vx(\theta)[a]
       \left(
          \frac{
             4\eta
             \bigl(\ut[a]-\vec u^{(t-1)}[a]\bigr)
          }{
             1+4\eta
             \bigl(
                (1-\theta)\vec u^{(t-1)}[a]+\theta\ut[a]
             \bigr)
          }
       \right)^2.
\]
The denominator is at least $3/4$, so
\[
    \left\|
       \frac{d}{d\theta}\sqrt{\vx(\theta)}
    \right\|_2^2
    \leq
    \frac{64\eta^2}{9}
    \norm{\ut-\vec u^{(t-1)}}_\infty^2.
\]
Finally, \eqref{eq:path-integral} yields
\[
    \norm{\sqrt{\vx(1)}-\sqrt{\vx(0)}}_2^2
    \leq
    \frac{64\eta^2}{9}
    \norm{\ut-\vec u^{(t-1)}}_\infty^2
    \leq
    8\eta^2
    \norm{\ut-\vec u^{(t-1)}}_\infty^2.
\]
This proves \eqref{eq:correction-movement}.
\end{proof}

\subsection{Changing the Cumulative Input}
\label{app:movement-input}

We next keep the optimistic correction fixed at $\ut$ and change the
potential input from $\eta\vec U^{(t)}$ to
$\eta\vec U^{(t+1)}$. Our goal is to control the resulting strategy
movement by the same gradient-weighted square that appears in
\eqref{eq:budget}. The compensation bound from
\Cref{lem:compensation} is what preserves these gradient weights after
normalization.

\begin{lemma}\label{lem:input-change}
Fix a round $t$ and $0<\eta\leq1/16$. For $0\leq\theta\leq1$,
let
\[
    \vx(\theta)[a]
      =
      \frac{
         a_c(\eta\vec U^{(t)}[a]+\theta\eta\ut[a])
         (1+4\eta\ut[a])
      }{
         \sum_b
         a_c(\eta\vec U^{(t)}[b]+\theta\eta\ut[b])
         (1+4\eta\ut[b])
      }.
\]
Then
\[
    \norm{\sqrt{\vx(1)}-\sqrt{\vx(0)}}_2^2
    \leq
    2\eta^2
    \sum_a
       \partial_a\Psi(\eta\vec U^{(t)})\ut[a]^2.
    \numberthis{eq:score-movement}
\]
\end{lemma}

\begin{proof}
The path starts at the intermediate distribution from the proof of
\Cref{lem:movement}, which uses the old cumulative input and the new
correction. At the other endpoint,
\[
    \eta\vec U^{(t)}+\eta\ut
      =
      \eta\vec U^{(t+1)},
\]
so $\vx(1)=\vx^{(t+1)}$.

Since $\norm{\ut}_\infty\leq1$ and $\eta\leq1/16$, the fixed correction
satisfies
\[
    \frac34
    \leq
    1+4\eta\ut[a]
    \leq
    \frac54.
\]
Thus all weights along the path are strictly positive.

We apply the normalized-path identity from
\Cref{lem:normalized-path}. Away from zero crossings, the logarithmic
derivative of the unnormalized weight of action $a$ is
\[
    \eta\ut[a]
    \frac{
       a_c'(\eta\vec U^{(t)}[a]+\theta\eta\ut[a])
    }{
       a_c(\eta\vec U^{(t)}[a]+\theta\eta\ut[a])
    }.
\]
Bounding the variance in \eqref{eq:normalized-speed} by its second
moment therefore gives
\[
    \left\|
       \frac{d}{d\theta}\sqrt{\vx(\theta)}
    \right\|_2^2
    \leq
    \frac{\eta^2}{4}
    \sum_a
       \vx(\theta)[a]
       \left(
          \frac{
             a_c'(\eta\vec U^{(t)}[a]+\theta\eta\ut[a])
          }{
             a_c(\eta\vec U^{(t)}[a]+\theta\eta\ut[a])
          }
       \right)^2
       \ut[a]^2.
\]

We now compare the corrected probabilities with the normalized gradient
coordinates. Using the bounds $3/4$ and $5/4$ on the correction,
\[
    \vx(\theta)[a]
    \leq
    \frac53
    \frac{
       a_c(\eta\vec U^{(t)}[a]+\theta\eta\ut[a])
    }{
       \sum_b
       a_c(\eta\vec U^{(t)}[b]+\theta\eta\ut[b])
    }.
\]
By the gradient formula in \Cref{lem:calculus}, the common factor in
$\partial_a\Psi$ cancels under normalization. Hence,
\[
    \vx(\theta)[a]
    \leq
    \frac53
    \frac{
       \partial_a\Psi(\eta\vec U^{(t)}+\theta\eta\ut)
    }{
       \sum_b
       \partial_b\Psi(\eta\vec U^{(t)}+\theta\eta\ut)
    }.
\]
Substituting this bound into the square-root speed gives
\[
    \left\|
       \frac{d}{d\theta}\sqrt{\vx(\theta)}
    \right\|_2^2
    \leq
    \frac{5\eta^2}{12}
    \sum_a
       \partial_a\Psi(\eta\vec U^{(t)}+\theta\eta\ut)
       \frac{
          \left(
             a_c'(\eta\vec U^{(t)}[a]+\theta\eta\ut[a])
             /
             a_c(\eta\vec U^{(t)}[a]+\theta\eta\ut[a])
          \right)^2
       }{
          \sum_b
          \partial_b\Psi(\eta\vec U^{(t)}+\theta\eta\ut)
       }
       \ut[a]^2.
\]
The ratio in each summand is exactly the quantity controlled by
\eqref{eq:compensation}. Applying that bound yields
\[
    \left\|
       \frac{d}{d\theta}\sqrt{\vx(\theta)}
    \right\|_2^2
    \leq
    \frac{5\eta^2}{4}
    \sum_a
       \partial_a\Psi(\eta\vec U^{(t)}+\theta\eta\ut)
       \ut[a]^2.
\]

Integrating with \eqref{eq:path-integral} gives
\[
    \norm{\sqrt{\vx(1)}-\sqrt{\vx(0)}}_2^2
    \leq
    \frac{5\eta^2}{4}
    \int_0^1
       \sum_a
          \partial_a\Psi(\eta\vec U^{(t)}+\theta\eta\ut)
          \ut[a]^2
    \,d\theta.
\]
Finally, \eqref{eq:segment} gives
\[
    \partial_a\Psi(\eta\vec U^{(t)}+\theta\eta\ut)
    \leq
    \frac43
    \partial_a\Psi(\eta\vec U^{(t)}),
\]
because $\norm{\ut}_\infty\leq1$ and $\eta\leq1/16$. Therefore,
\[
    \norm{\sqrt{\vx(1)}-\sqrt{\vx(0)}}_2^2
    \leq
    \frac{5\eta^2}{3}
    \sum_a
       \partial_a\Psi(\eta\vec U^{(t)})\ut[a]^2
    \leq
    2\eta^2
    \sum_a
       \partial_a\Psi(\eta\vec U^{(t)})\ut[a]^2.
\]
This proves \eqref{eq:score-movement}. The derivative calculations hold
almost everywhere along the path. As noted after
\Cref{lem:compensation}, each nonconstant coordinate crosses zero at
most once, so these exceptional points do not affect the integral.
\end{proof}

\subsection{Combining the Correction and Input Changes}
\label{app:movement-combination}

The two paths from the preceding subsections meet at the same
intermediate distribution. We now use this common endpoint to combine
the correction-change and cumulative-input bounds and recover the
strategy-movement estimate from the main analysis.

\begin{repeatlemma}{lem:movement}
Suppose player $i$ uses \algname{} with a fixed rate
$0<\eta\leq1/16$ and receives any sequence of utility vectors in
$[0,1]^{|\mathcal A_i|}$ after play. For every $t\geq1$,
\[
    \norm{\sqrt{\vx_i^{(t+1)}}-\sqrt{\xt_i}}_2^2
      \leq4\eta^2\sum_a
         \partial_a\Psi_i(\eta\vec U_i^{(t)})\ut_i[a]^2
        +16\eta^2\norm{\ut_i-\vec u_i^{(t-1)}}_\infty^2.
\]
\end{repeatlemma}

\begin{proof}
Let $\widetilde{\vx}_i$ denote the intermediate distribution with
probabilities
\[
    \widetilde{\vx}_i[a]
      =
      \frac{
         a_c(\eta\vec U_i^{(t)}[a])
         (1+4\eta\ut_i[a])
      }{
         \sum_b
         a_c(\eta\vec U_i^{(t)}[b])
         (1+4\eta\ut_i[b])
      }.
\]
The correction path from \Cref{lem:correction-change} starts at
$\xt_i$ and ends at $\widetilde{\vx}_i$. Hence,
\[
    \norm{\sqrt{\widetilde{\vx}_i}-\sqrt{\xt_i}}_2^2
    \leq
    8\eta^2
    \norm{\ut_i-\vec u_i^{(t-1)}}_\infty^2.
\]
The cumulative-input path from \Cref{lem:input-change} starts at
$\widetilde{\vx}_i$. Since
$\vec U_i^{(t+1)}=\vec U_i^{(t)}+\ut_i$, its endpoint is
$\vx_i^{(t+1)}$. Therefore,
\[
    \norm{\sqrt{\vx_i^{(t+1)}}-\sqrt{\widetilde{\vx}_i}}_2^2
    \leq
    2\eta^2
    \sum_a
       \partial_a\Psi_i(\eta\vec U_i^{(t)})\ut_i[a]^2.
\]

The triangle inequality for the square-root vectors gives
\[
    \norm{\sqrt{\vx_i^{(t+1)}}-\sqrt{\xt_i}}_2
    \leq
    \norm{\sqrt{\vx_i^{(t+1)}}-\sqrt{\widetilde{\vx}_i}}_2
    +
    \norm{\sqrt{\widetilde{\vx}_i}-\sqrt{\xt_i}}_2.
\]
Squaring and using $(r+s)^2\leq2r^2+2s^2$, followed by the two bounds
above, yields
\[
    \norm{\sqrt{\vx_i^{(t+1)}}-\sqrt{\xt_i}}_2^2
    \leq
    4\eta^2
    \sum_a
       \partial_a\Psi_i(\eta\vec U_i^{(t)})\ut_i[a]^2
    +
    16\eta^2
    \norm{\ut_i-\vec u_i^{(t-1)}}_\infty^2.
\]
This proves the lemma. The intermediate distribution
$\widetilde{\vx}_i$ is used only to separate the two changes in the
response and is never played by the algorithm.
\end{proof}

\section{Hellinger Distance and Prediction Error}\label{app:game}

The proof of \Cref{lem:game} uses two properties of the square-root
distance underlying Hellinger distance. First, it controls changes in
expectations of bounded functions. Second, for product distributions it factorizes nicely, i.e.,
its square is at most the sum of the squared distances between the
factors. The second property is what avoids an additional factor in the
number of opponents.

We prove these two facts below. The game-specific application, including
the effect of centering the utility vectors, is carried out in
\Cref{sec:analysis-game}.

\subsection{Hellinger Distance and Differences of Expectations}
\label{app:game-expectations}

We first relate square-root distance to total variation. This immediately
gives the expectation bounds used in the fixed-game analysis.

\begin{lemma}\label{lem:expectation-comparison}
Let $\vx,\vx'$ be probability vectors on the same finite action set.
Then
\[
    \norm{\vx-\vx'}_1
    \leq
    2\norm{\sqrt{\vx}-\sqrt{\vx'}}_2.
    \numberthis{eq:root-tv}
\]
For a function $g$ on this set with values in $[-1,1]$,
\[
    \left|
       \E_{a\sim\vx}[g(a)]
       -
       \E_{a\sim\vx'}[g(a)]
    \right|
    \leq
    2\norm{\sqrt{\vx}-\sqrt{\vx'}}_2.
\]
If $g$ takes values in $[0,1]$, the factor $2$ can be replaced by $1$.
\end{lemma}

\begin{proof}
For each action,
\[
    \vx[a]-\vx'[a]
      =
      \bigl(\sqrt{\vx[a]}-\sqrt{\vx'[a]}\bigr)
      \bigl(\sqrt{\vx[a]}+\sqrt{\vx'[a]}\bigr).
\]
Hence, by Cauchy--Schwarz,
\[
    \norm{\vx-\vx'}_1
    \leq
    \norm{\sqrt{\vx}-\sqrt{\vx'}}_2
    \norm{\sqrt{\vx}+\sqrt{\vx'}}_2.
\]
Since $\norm{\sqrt{\vx}}_2=\norm{\sqrt{\vx'}}_2=1$, the triangle
inequality gives
\[
    \norm{\sqrt{\vx}+\sqrt{\vx'}}_2
    \leq2,
\]
which proves \eqref{eq:root-tv}. Notice that the argument also allows
zero-probability coordinates.

Now suppose $g$ takes values in $[-1,1]$. Then
\[
    \left|
       \E_{a\sim\vx}[g(a)]
       -
       \E_{a\sim\vx'}[g(a)]
    \right|
    =
    \left|
       \sum_a(\vx[a]-\vx'[a])g(a)
    \right|
    \leq
    \norm{\vx-\vx'}_1.
\]
Combining this with \eqref{eq:root-tv} gives the factor $2$.

If instead $g$ takes values in $[0,1]$, then
$\sum_a(\vx[a]-\vx'[a])=0$, so subtracting $1/2$ from $g$ does not
change the expectation difference. Therefore,
\[
    \left|
       \E_{a\sim\vx}[g(a)]
       -
       \E_{a\sim\vx'}[g(a)]
    \right|
    =
    \left|
       \sum_a
          (\vx[a]-\vx'[a])
          \left(g(a)-\frac12\right)
    \right|
    \leq
    \frac12\norm{\vx-\vx'}_1.
\]
Applying \eqref{eq:root-tv} once more gives the stated factor $1$.
\end{proof}

The two versions of the expectation bound are both used in
\Cref{lem:game}. The $[-1,1]$ bound controls changes in payoff
differences caused by the opponents, while the $[0,1]$ bound controls
the player's centering term.

\subsection{Hellinger Distance between Product Distributions}
\label{app:game-products}

The second property concerns product distributions. The square-root
overlap factors across independent coordinates, which allows the
squared distance between two products to be controlled by the sum of
the squared distances between their factors.

\begin{lemma}\label{lem:product-comparison}
Consider any finite collection of pairs $\vx_j,\vx'_j$ of probability
vectors, where each pair is defined on the same finite set
$\mathcal A_j$. Then
\[
    \left\|
       \sqrt{\bigotimes_j\vx_j}
       -
       \sqrt{\bigotimes_j\vx'_j}
    \right\|_2^2
    \leq
    \sum_j
       \norm{\sqrt{\vx_j}-\sqrt{\vx'_j}}_2^2.
    \numberthis{eq:product}
\]
All products and sums are over the same collection. The statement
also holds for the empty collection, with both sides equal to zero.
\end{lemma}

\begin{proof}
For each factor $j$, define its square-root overlap by
\[
    h_j
      =
      \sum_{a\in\mathcal A_j}
         \sqrt{\vx_j[a]\vx'_j[a]}.
\]
Cauchy--Schwarz gives $0\leq h_j\leq1$. Expanding the squared
square-root distance gives
\[
    \norm{\sqrt{\vx_j}-\sqrt{\vx'_j}}_2^2
      =
      2(1-h_j).
\]

The corresponding overlap of the product distributions factors into
the product of the individual overlaps. Indeed,
\[
    \sum_{\vec s\in\prod_j\mathcal A_j}
       \prod_j
          \sqrt{\vx_j[\vec s[j]]\vx'_j[\vec s[j]]}
    =
    \prod_j h_j.
\]
Therefore,
\[
    \left\|
       \sqrt{\bigotimes_j\vx_j}
       -
       \sqrt{\bigotimes_j\vx'_j}
    \right\|_2^2
    =
    2\left(1-\prod_j h_j\right).
\]

It remains to compare the product overlap with the individual ones.
Since every $h_j$ belongs to $[0,1]$,
\[
    1-\prod_j h_j
    \leq
    \sum_j(1-h_j).
\]
For example, this follows by expanding $1-\prod_j h_j$ successively
and observing that every remaining product of overlaps is at most one.
Multiplying by two gives
\[
    2\left(1-\prod_j h_j\right)
    \leq
    \sum_j2(1-h_j),
\]
which is exactly \eqref{eq:product}.

For the empty collection, both product distributions are the point
mass on the empty profile, so both sides are zero.
\end{proof}

The two lemmas together give the comparison used in the fixed-game
analysis. If $g$ takes values in $[-1,1]$, then
\[
    \left|
       \E_{\vec s\sim\bigotimes_j\vx_j}[g(\vec s)]
       -
       \E_{\vec s\sim\bigotimes_j\vx'_j}[g(\vec s)]
    \right|
    \leq
    2
    \left(
       \sum_j
          \norm{\sqrt{\vx_j}-\sqrt{\vx_j'}}_2^2
    \right)^{1/2}.
\]
In \Cref{lem:game}, the factors are the opponents' strategies and
\[
    g(\vec s_{-i})
      =
      \mathcal U_i(a,\vec s_{-i})
      -
      \mathcal U_i(b,\vec s_{-i}).
\]
Thus, the change in a payoff difference is controlled directly by the
sum of the opponents' squared Hellinger movements. 

\section{Adversarial Regret with a Learning-Rate Safeguard}
\label{app:safeguard}

We now establish the adversarial guarantee by adding a learning-rate
safeguard to \algname{}. The potential, multiplicative correction,
centered-utility update, and cumulative vector remain unchanged. Only
the learning rate is allowed to decrease.

The construction is inspired by the adaptive-step-size argument of
\citet[Appendix~D]{daskalakis2021near}, which monitors a condition
satisfied under self-play and adjusts the learning rate when that
condition fails. Here we monitor the potential bound obtained in the
fixed-game analysis. If the potential crosses the threshold $4c$, we
decrease the learning rate before the next round. There is no restart
and no switch to a different learning rule.

The analysis has two parts. For an arbitrary sequence of utility
vectors, we show that the learning-rate adjustment keeps
\[
    \Psi_i(\eta_i^{(t)}\vec U_i^{(t)})\leq4c
\]
on every round. We then show that
$1/(\eta_i^{(t)})^2$ increases by at most $49/c$ per round, which
prevents the learning rate from becoming too small. In fixed-game
self-play, the potential never crosses the threshold, so the safeguard
never changes the learning rate and the trajectory remains exactly the
one analyzed in \Cref{sec:analysis}.

\subsection{The Safeguarded Update}
\label{app:safeguard-update}

Fix a player $i$, set $c=2+\log d$, and initialize
\[
    \vec U_i^{(1)}=\zero,
    \qquad
    \vec u_i^{(0)}=\zero,
    \qquad
    \eta_i^{(1)}=\frac1{32\sqrt n}.
\]
On round $t$, the player uses the current learning rate in both parts
of the \algname{} response,
\[
    \xt_i
      \propto
      \nabla\Psi_i(\eta_i^{(t)}\vec U_i^{(t)})
      \odot
      \bigl(\one+4\eta_i^{(t)}\vec u_i^{(t-1)}\bigr).
    \numberthis{eq:safeguard-rule}
\]
After observing $\nut_i$, the player performs the same centering and
cumulative updates as before,
\[
    \ut_i
      =
      \nut_i-\ip{\xt_i}{\nut_i}\one,
    \qquad
    \vec U_i^{(t+1)}
      =
      \vec U_i^{(t)}+\ut_i.
\]

The safeguard is applied after this update. We measure the amount by
which the potential, evaluated using the round-$t$ learning rate,
exceeds $4c$,
\[
    \delta_i^{(t)}
      =
      \bigl[
         \Psi_i(\eta_i^{(t)}\vec U_i^{(t+1)})-4c
      \bigr]_+,
\]
and choose
\[
    \eta_i^{(t+1)}
      =
      \frac{\eta_i^{(t)}}
           {1+\delta_i^{(t)}/c}.
    \numberthis{eq:safeguard-rate}
\]
Thus the learning rate remains unchanged whenever the potential stays
below $4c$ and decreases only after the threshold is crossed.

The potential defining $\delta_i^{(t)}$ uses the round-$t$ rate and
the cumulative vector after incorporating the round-$t$ utility. Any
rate decrease takes effect only on round $t+1$. 

\subsection{The Adversarial Regret Bound}
\label{app:safeguard-bound}

\begin{theorem}[Adversarial regret]\label{thm:adversarial}
Fix $n\geq1$, $d\geq2$, and a player $i$ with
$1\leq|\mathcal A_i|\leq d$. With $c=2+\log d$ and initial rate
$\eta_i^{(1)}=1/(32\sqrt n)$, the rule
\eqref{eq:safeguard-rule}--\eqref{eq:safeguard-rate} satisfies
\[
    \Reg_i^{(T)}
      \leq
      96\sqrt n\,(2+\log d)
      +21\sqrt{T(2+\log d)}
\]
for every sequence of utility vectors
$\nut_i\in[0,1]^{|\mathcal A_i|}$ observed in full after play and
every integer $T\geq1$. The utility vectors may be chosen adaptively,
with no assumption on the other players' behavior. The rule is
deterministic, uncoupled, and independent of the horizon.
\end{theorem}

\begin{proof}
\emph{Positivity of the learning rates and response weights.}
The excess $\delta_i^{(t)}$ is finite and nonnegative at every finite
time. Hence, \eqref{eq:safeguard-rate} gives
\[
    0
    <
    \eta_i^{(t+1)}
    \leq
    \eta_i^{(t)}
    \leq
    \frac1{32\sqrt n}
    \leq
    \frac1{32}.
\]
Moreover, centering gives $\norm{\ut_i}_\infty\leq1$. Therefore every
optimistic correction factor satisfies
\[
    1+4\eta_i^{(t)}\vec u_i^{(t-1)}[a]
    \geq
    1-4\eta_i^{(t)}
    \geq
    \frac78.
\]
Together with the positivity of the gradient coordinates from
\Cref{lem:calculus}, this shows that the response is well defined on
every round.

\emph{Keeping the potential below the threshold.}
We first prove
\[
    \Psi_i(\eta_i^{(t)}\vec U_i^{(t)})
    \leq
    4c
    \qquad\text{for every }t\geq1.
    \numberthis{eq:safeguard-cap}
\]
At $t=1$, the input is zero and \Cref{lem:calculus} gives
$\Psi_i(\zero)<3c<4c$.

Suppose \eqref{eq:safeguard-cap} holds at the beginning of round $t$.
If $\delta_i^{(t)}=0$, then
\[
    \Psi_i(\eta_i^{(t)}\vec U_i^{(t+1)})
    \leq
    4c
\]
and the learning rate does not change. Hence
\eqref{eq:safeguard-cap} also holds at $t+1$.

Now suppose $\delta_i^{(t)}>0$. By its definition,
\[
    \Psi_i(\eta_i^{(t)}\vec U_i^{(t+1)})
      =
      4c+\delta_i^{(t)}.
\]
The rate update gives
\[
    \eta_i^{(t+1)}\vec U_i^{(t+1)}
      =
      \frac{c}{c+\delta_i^{(t)}}
      \eta_i^{(t)}\vec U_i^{(t+1)}
      +
      \frac{\delta_i^{(t)}}{c+\delta_i^{(t)}}\zero.
\]
Thus the new potential input lies on the segment between the old
scaled input and zero. By convexity of $\Psi_i$ and
$\Psi_i(\zero)<3c$,
\[
    \Psi_i(\eta_i^{(t+1)}\vec U_i^{(t+1)})
    \leq
    \frac{c}{c+\delta_i^{(t)}}(4c+\delta_i^{(t)})
    +
    \frac{\delta_i^{(t)}}{c+\delta_i^{(t)}}3c
    =
    4c.
\]
This proves \eqref{eq:safeguard-cap}. The argument uses convexity of
the potential along this segment and does not require the potential to
be monotone as the learning rate decreases.

\emph{Bounding the excess above the threshold.}
We next control how far the potential can move above $4c$ before the
rate adjustment is applied. Although the cumulative fixed-rate bound
\eqref{eq:budget} does not telescope when the learning rate changes,
the one-round estimate in \Cref{lem:one-round} remains valid with the
rate used on that round.

Apply \Cref{lem:one-round} with
$\eta=\eta_i^{(t)}$. Since \eqref{eq:safeguard-rule} is exactly the
corresponding one-round response,
\[
    \Psi_i(\eta_i^{(t)}\vec U_i^{(t+1)})
    -
    \Psi_i(\eta_i^{(t)}\vec U_i^{(t)})
    \leq
    2(\eta_i^{(t)})^2
    \sum_a
       \partial_a\Psi_i(\eta_i^{(t)}\vec U_i^{(t)})
       \bigl(\ut_i[a]-\vec u_i^{(t-1)}[a]\bigr)^2
\]
\[
    \hspace{4em}
    -
    (\eta_i^{(t)})^2
    \sum_a
       \partial_a\Psi_i(\eta_i^{(t)}\vec U_i^{(t)})
       \ut_i[a]^2.
\]
The last term is nonpositive. Using
$\sum_a\partial_a\Psi_i\leq3$ from \Cref{lem:calculus} therefore gives
\[
    \Psi_i(\eta_i^{(t)}\vec U_i^{(t+1)})
    -
    \Psi_i(\eta_i^{(t)}\vec U_i^{(t)})
    \leq
    6(\eta_i^{(t)})^2
    \norm{\ut_i-\vec u_i^{(t-1)}}_\infty^2.
\]
Both centered-utility vectors have infinity norm at most one, so
\[
    \norm{\ut_i-\vec u_i^{(t-1)}}_\infty
    \leq
    2.
\]
Consequently,
\[
    \Psi_i(\eta_i^{(t)}\vec U_i^{(t+1)})
    -
    \Psi_i(\eta_i^{(t)}\vec U_i^{(t)})
    \leq
    24(\eta_i^{(t)})^2.
\]
Together with \eqref{eq:safeguard-cap}, this yields
\[
    0
    \leq
    \delta_i^{(t)}
    \leq
    24(\eta_i^{(t)})^2.
\]

\emph{Controlling the decrease of the learning rate.}
The preceding estimate controls the rate update through the reciprocal
squared learning rate. From \eqref{eq:safeguard-rate},
\[
    \frac1{(\eta_i^{(t+1)})^2}
    -
    \frac1{(\eta_i^{(t)})^2}
    =
    \frac{2\delta_i^{(t)}}
         {c(\eta_i^{(t)})^2}
    +
    \frac{(\delta_i^{(t)})^2}
         {c^2(\eta_i^{(t)})^2}.
\]
Using $\delta_i^{(t)}\leq24(\eta_i^{(t)})^2$ gives
\[
    \frac1{(\eta_i^{(t+1)})^2}
    -
    \frac1{(\eta_i^{(t)})^2}
    \leq
    \frac{48}{c}
    +
    \frac{576(\eta_i^{(t)})^2}{c^2}.
\]
Since $\eta_i^{(t)}\leq1/32$ and $c>2$,
\[
    576(\eta_i^{(t)})^2
    \leq
    \frac{576}{1024}
    <
    c,
\]
and hence
\[
    \frac1{(\eta_i^{(t+1)})^2}
    -
    \frac1{(\eta_i^{(t)})^2}
    \leq
    \frac{49}{c}.
\]
Summing from $t=1$ to $T$ telescopes these reciprocal-rate differences,
which gives
\[
    \frac1{(\eta_i^{(T+1)})^2}
    \leq
    \frac1{(\eta_i^{(1)})^2}
    +
    \frac{49T}{c}
    =
    1024n+\frac{49T}{c}.
\]
Thus the safeguard can decrease the learning rate only at the rate
allowed by this finite-horizon bound.

\emph{Converting the potential bound into regret.}
After round $T$, \eqref{eq:safeguard-cap} gives
\[
    \Psi_i(\eta_i^{(T+1)}\vec U_i^{(T+1)})
    \leq
    4c.
\]
Applying the potential certificate from \Cref{lem:calculus},
\[
    c
    +
    \eta_i^{(T+1)}
    \max_a\vec U_i^{(T+1)}[a]
    \leq
    \Psi_i(\eta_i^{(T+1)}\vec U_i^{(T+1)})
    \leq
    4c.
\]
Using \eqref{eq:two-identities},
\[
    \Reg_i^{(T)}
    \leq
    \frac{3c}{\eta_i^{(T+1)}}.
\]
The reciprocal-rate bound now gives
\[
    \Reg_i^{(T)}
    \leq
    3c
    \sqrt{
       1024n+\frac{49T}{c}
    }
    \leq
    96\sqrt n\,c
    +
    21\sqrt{cT}.
\]
Substituting $c=2+\log d$ proves the stated regret bound.

All inequalities above hold pathwise for every realized sequence of
utility vectors. No independence or probabilistic assumption on the
sequence is used, so the result also allows the utility vectors to be
chosen adaptively from the preceding history. The update uses only
quantities available by the end of the current round and does not
depend on $T$.
\end{proof}

\subsection{The Safeguard under Self-play}
\label{app:safeguard-selfplay}

We finally show that the safeguard is inactive in the fixed-game
self-play setting. Thus the adversarial protection does not alter the
trajectory used to obtain the stronger self-play guarantee.

\begin{proposition}
\label{prop:safeguard-selfplay}
Suppose all players use
\eqref{eq:safeguard-rule}--\eqref{eq:safeguard-rate}
in the fixed-game setting of \Cref{thm:self-play}. Then
\[
    \eta_i^{(t)}=\frac1{32\sqrt n}
    \qquad\text{for every player }i\text{ and every }t\geq1.
\]
Their mixed strategies coincide, round by round, with those of the
fixed-rate algorithm. The regret and equilibrium bounds in
\Cref{thm:self-play,cor:cce} therefore remain unchanged.
\end{proposition}

\begin{proof}
We argue by induction on the round. Initially, every player has the
same learning rate, cumulative vector, and preceding centered utility
as in the fixed-rate algorithm. Hence the first-round strategies
coincide. Suppose all players have retained the rate $1/(32\sqrt n)$ through
round $t$. Their played strategies and observations through that round
then coincide with those of fixed-rate self-play, so their updated
cumulative vectors $\vec U_i^{(t+1)}$ also coincide.

Apply \Cref{prop:closure} to this finite prefix. For every player $i$,
\[
    \Psi_i\!\left(
       \frac{\vec U_i^{(t+1)}}{32\sqrt n}
    \right)
    \leq
    4c.
\]
Therefore $\delta_i^{(t)}=0$ for every player, and
\eqref{eq:safeguard-rate} gives
\[
    \eta_i^{(t+1)}
      =
      \frac1{32\sqrt n}.
\]
This proves the induction step simultaneously for all players. Hence
the safeguard never changes a learning rate, and the entire self-play
trajectory agrees with the fixed-rate trajectory.
\end{proof}

The safeguard requires only one additional scalar learning rate for each
player and one evaluation of the closed-form potential after each
round. It does not change the cumulative vector or centered utilities,
introduce an additional play, or require an optimization subroutine.

\subsection{Completing the Main Guarantee}

The self-play result and the adversarial bound together give the main
theorem.

\begin{repeatthm}{thm:main}[Regret bounds for \algname{}]
Consider a fixed $n$-player finite game with at most $d\geq2$ actions
per player and utilities in $[0,1]$. If all players $i\in[n]$ follow
\algname{} (\Cref{alg:learning}) with $c=2+\log d$ and learning rate
$\eta=1/(32\sqrt n)$, observing their exact expected-utility vectors
after each round, then every player $i\in[n]$ has external regret
\[
    \Reg_i^{(T)}\leq \mathcal{O}(\sqrt n\log d).
\]
Moreover, \algname{} is adaptive to adversarial utilities through the
learning-rate safeguard in \Cref{app:safeguard}, which only decreases
the learning rate adaptively and leaves the rest of the update unchanged. The safeguarded rule for any individual
player $i$ guarantees
\[
    \Reg_i^{(T)}\leq \mathcal{O}(\sqrt n\log d+\sqrt{T\log d})
\]
against arbitrary, possibly adaptive, utility vectors
$\nut_i\in[0,1]^{|\mathcal A_i|}$. Both bounds hold uniformly over all $T\geq1$.
\end{repeatthm}

\begin{proof}
The self-play assertion follows from \Cref{thm:self-play}, which gives
the explicit bound
\[
    \Reg_i^{(T)}
    \leq
    96\sqrt n\,(2+\log d).
\]
The adversarial assertion follows from \Cref{thm:adversarial}, which
gives
\[
    \Reg_i^{(T)}
    \leq
    96\sqrt n\,(2+\log d)
    +
    21\sqrt{T(2+\log d)}.
\]
For $d\geq2$,
\[
    2+\log d
    \leq
    \left(1+\frac2{\log2}\right)\log d,
\]
so these two explicit estimates imply the stated asymptotic bounds
with universal constants. Both hold for every finite horizon without
requiring $T$ in advance. Finally,
\Cref{prop:safeguard-selfplay} shows that the safeguard never changes
the learning rate in fixed-game self-play, so the stronger self-play
trajectory and its equilibrium guarantee remain unchanged.
\end{proof}

\end{document}

%% file: notation.tex
\def\[#1\]{\begin{align*}#1\end{align*}}

\NewDocumentCommand{\numberthis}{om}{%
  \IfNoValueTF{#1}{%
    \refstepcounter{equation}\tag{\theequation}%
  }{%
    \tag{#1}%
  }%
  \label{#2}%
}

\newtheorem{theorem}{Theorem}[section]

\newaliascnt{lemma}{theorem}
\newtheorem{lemma}[lemma]{Lemma}
\aliascntresetthe{lemma}

\newaliascnt{proposition}{theorem}
\newtheorem{proposition}[proposition]{Proposition}
\aliascntresetthe{proposition}

\newaliascnt{corollary}{theorem}
\newtheorem{corollary}[corollary]{Corollary}
\aliascntresetthe{corollary}

\newaliascnt{claim}{theorem}

\aliascntresetthe{claim}

\theoremstyle{definition}
\newaliascnt{definition}{theorem}

\aliascntresetthe{definition}

\newaliascnt{question}{theorem}

\aliascntresetthe{question}

\theoremstyle{plain}
\newaliascnt{remark}{theorem}
\newtheorem{remark}[remark]{Remark}
\aliascntresetthe{remark}

\crefname{theorem}{theorem}{theorems}
\Crefname{theorem}{Theorem}{Theorems}
\crefname{lemma}{lemma}{lemmas}
\Crefname{lemma}{Lemma}{Lemmas}
\crefname{proposition}{proposition}{propositions}
\Crefname{proposition}{Proposition}{Propositions}
\crefname{corollary}{corollary}{corollaries}
\Crefname{corollary}{Corollary}{Corollaries}
\crefname{definition}{definition}{definitions}
\Crefname{definition}{Definition}{Definitions}
\crefname{section}{section}{sections}
\Crefname{section}{Section}{Sections}

\newcommand{\R}{\mathbb{R}}

\newcommand{\eps}{\varepsilon}
\newcommand{\one}{\mathbf{1}}

\newcommand{\norm}[1]{\left\lVert #1\right\rVert}

\numberwithin{equation}{section}

\theoremstyle{plain}

\newcommand{\repeatthmname}{}
\newtheorem*{repeatthminner}{\repeatthmname}
\newenvironment{repeatthm}[1]{%
  \renewcommand{\repeatthmname}{Theorem~\ref{#1}}%
  \begin{repeatthminner}%
}{%
  \end{repeatthminner}%
}

\newcommand{\repeatlemmaname}{}
\newtheorem*{repeatlemmainner}{\repeatlemmaname}
\newenvironment{repeatlemma}[1]{%
  \renewcommand{\repeatlemmaname}{Lemma~\ref{#1}}%
  \begin{repeatlemmainner}%
}{%
  \end{repeatlemmainner}%
}

\newcommand{\repeatpropname}{}
\newtheorem*{repeatpropinner}{\repeatpropname}

\newcommand{\repeatcorollaryname}{}
\newtheorem*{repeatcorollaryinner}{\repeatcorollaryname}
\newenvironment{repeatcorollary}[1]{%
  \renewcommand{\repeatcorollaryname}{Corollary~\ref{#1}}%
  \begin{repeatcorollaryinner}%
}{%
  \end{repeatcorollaryinner}%
}

\renewcommand{\vec}[1]{\bm{#1}}
\newcommand{\vx}{\vec{x}}
\newcommand{\nut}{\vec{\nu}^{(t)}}
\newcommand{\xt}{\vx^{(t)}}
\newcommand{\ut}{\vec{u}^{(t)}}
\newcommand{\zero}{\vec{0}}
\renewcommand{\one}{\vec{1}}
\newcommand{\Reg}{\operatorname{Reg}}
\newcommand{\E}{\mathbb{E}}
\newcommand{\ip}[2]{\left\langle #1,#2\right\rangle}
\DeclareMathOperator{\diag}{diag}
\DeclareMathOperator{\Var}{Var}

\newcommand{\algname}{\texttt{MORM}}